\documentclass[10pt]{article}
\pdfoutput=1
\usepackage[backend=bibtex]{biblatex}
\usepackage[pdftex]{graphicx}
\usepackage{amsmath}
\usepackage{algorithmicx}
\usepackage{amsfonts}
\usepackage{amsthm}
\usepackage{algpseudocode}
\usepackage{array}
\usepackage{color}
\usepackage{hyperref}
\usepackage{url}
\usepackage{fullpage}
\usepackage{parskip}
\usepackage{xspace}
\usepackage[T1]{fontenc}
\usepackage{algorithm}
\usepackage{todonotes}
\usepackage{caption}

\algtext*{EndWhile}
\algtext*{EndIf}
\algtext*{EndProcedure}
\algtext*{EndFor}   

\newcounter{algsubstate}
\renewcommand{\thealgsubstate}{\alph{algsubstate}}

\newtheorem{lemma}{Lemma}

\graphicspath{{./figures/}}

\newcommand{\paran}[1]{\left( #1 \right)}
\newcommand{\remoteGate}[0]{\textsc{RemoteGate}\xspace}

\begin{document}

\title{\textsc{RemoteGate}: Incentive-Compatible Remote Configuration of Security Gateways}

\author{Abhinav Aggarwal$^{*,\dagger}$, Mahdi Zamani$^*$, Mihai Christodorescu$^*$ \\
\textit{$^*$Visa Research, Palo Alto CA}\\
\textit{$^\dagger$University of New Mexico, Albuquerque, NM}
}
%\author{}
\date{}
\maketitle

\begin{abstract}
Imagine that a malicious hacker is trying to attack a server over the Internet and the server wants to block the attack packets as close to their point of origin as possible. However, the security gateway ahead of the source of attack is untrusted. How can the server block the attack packets through this gateway? In this paper, we introduce \remoteGate, a trustworthy mechanism for allowing any party (server) on the Internet to configure a security gateway owned by a second party, at a certain agreed upon reward that the former pays to the latter for its service. We take an interactive incentive-compatible approach, for the case when both the server and the gateway are rational, to devise a protocol that will allow the server to help the security gateway generate and deploy a policy rule that filters the attack packets before they reach the server. The server will reward the gateway only when the latter can successfully verify that it has generated and deployed the correct rule for the issue. This mechanism will enable an Internet-scale approach to improving security and privacy, backed by digital payment incentives.
\end{abstract}

%\tableofcontents

\section{Introduction} 
% \section{Introduction}
\label{sec:intro}

% overview of the section
% What are you trying to do? Articulate your objectives using absolutely no jargon.
Enterprise and home networks are nowadays connected to the Internet via security gateways, such as firewalls, intrusion detection systems, and data-leakage protection systems, designed to protect the nodes on the local network from attacks originating on the Internet. This network architecture is widely prevalent, to the point of every home-network router having a built-in firewall. We propose to add another purpose to these security gateways, by tasking them with protecting the Internet at large from attacks originating on the local network. The mechanism presented in this paper will allow any Internet party (henceforth, referred to as \emph{server}) to configure a security gateway for a local network not under their administrative control.

% How is it done today, and what are the limits of current practice?
Current gateways perform some filtering of outgoing traffic, primarily with the goals of preventing sensitive-data exfiltration and of detecting infected local nodes. For example, such filtering applies to outgoing traffic destined to or originated on certain ports (e.g., blocking known malware-related ports, blocking broadcast traffic) or with certain contents (e.g., blocking HTTP requests to known-bad servers). The administrator of the local network can manage such outgoing filtering, and does so with the goal of improving the security of the local network. Thus, it is often the case that the security of the rest of the Internet is not considered, leading to the common result of networked computers infected with bot malware for long periods of time because the bot code does not impede on the computer owner's use.

% What is new in your approach and why do you think it will be successful?
%% - remote party can configure the middlebox only for shaping traffic that directly affects them
Our goal in this work is to enable a server to configure any gateway to block attack traffic that reaches the server after passing through the gateway. We introduce \remoteGate, a protocol that enables this remote configuration of the gateway's filtering mechanism by the server under attack. This provides any gateway on the Internet a dual-purpose : (1) it will protect some local network from external attacks by blocking unwanted packets originating on the Internet; and, (2) it will protect the Internet from attacks by blocking unwanted packets originating from the local network. A significant fact about this design is that the parties most interested in being protected from attack are capable of setting the corresponding policy on the gateway.

% Who cares? If you succeed, what difference will it make?
%% - the whole Internet becomes an SDN, with a global security policy that reflects each Internet nodes' needs
We envision \remoteGate to enable the creation of a distributed, Internet-scale packet filtering system. For each server, there would be one or more (security) gateways that mediate its interactions with the rest of the Internet. These gateways are enforcing a policy that was jointly specified by Internet parties interested in traffic to and from this server. From the point of view of a attack victim (the server in this case), this system offloads any defense mechanism from the local security gateway (closest to the victim) to a security gateway located closer to the source of the attack, reducing the processing load on the victim's security gateway and the traffic load on the core of the Internet (by stopping attack traffic closer to source). An additional advantage is the increase in accuracy at each security gateway of filtering the traffic outgoing from its local network, in addition to monitoring the incoming traffic (which can be arbitrary). 

% What are the risks?
%% - carelessly designed system would open new attack vectors into local networks
%% 4 threats:
%% 1) reduces security of local network
%% 2) allows DoS on the local network's Internet connectivity
%% 3) reduces security of server, because server may reduce its own local security policy in expectation that the gateway security policy would protect it
%% 4) allows DoS on the server, by preventing honest clients from reaching the server 
To facilitate \remoteGate as a system that allows the policy of any security gateway to be manipulated, several significant security challenges have to be addressed to make our proposal practical. First and foremost, an adversary could weaken or even eliminate the current filtering policy of the gateway through careful attack packet design, leaving the local network open to external attacks. In other words, it may seem possible for the server to make the gateway learn a model that affects the filtering already provided by the existing policies in its firewall. However, we emphasize that with \remoteGate in place, this scenario is completely prevented (see Appendix~\ref{app:technical} for details.) 

Second, a malicious update to the policy could prevent the local network from accessing Internet services it needs, effectively installing a denial-of-service policy. Third, if the server offloads its security enforcement to security gateways elsewhere on the Internet, a malicious security gateway could pretend to enforce the new policy while still allowing attacks to pass through, exposing the now-vulnerable server to exploitation. Fourth, an attacker with access to many gateways could manipulate them to block benign traffic destined for a server, again installing a denial-of-service policy, this time against the server.

Beyond the security issues potentially introduced by the mutual lack of trust, moving the security policy onto a local network's security gateway presents new challenges related to computational requirements and to deployment. If, say, Amazon.com decides to offload its security policy from its security firewall to the security firewalls of its top-100 million customers, then each of those customers will presumably have to prepare their firewalls to handle the increased traffic filtering task. Furthermore, if only some small percentage of those top-100 million customers are willing to participate in this \remoteGate system, Amazon.com may decide that the effort of offloading only a small part of their security firewall is not worth the small benefit to be gained. Thus, in our system we have to design incentives for both parties to participate, in addition to addressing the newly introduced security problems.

% How much will it cost?
%% - N/A

% How long will it take?
%% - N/A

% What are the mid-term and final “exams” to check for success?
%% - N/A

Our approach to building this \remoteGate system is three-fold. First, policy updates must apply only to traffic destined for the server requesting the policy change. Second, servers will pay gateways for updating their policies. Third, servers will check through an interactive protocol that gateways enforce the new policy. Underlying these requirements is a common negotiation substrate through which the server conveys to the gateway the desired update to the security policy, the gateway demonstrates to the server that it deployed the policy update, and payment from server to gateway is exchanged. These steps are not trivial as they rely on a number of assumptions that we aim to discard through careful design. We overcome these challenges by proposing a \remoteGate architecture, in which any Internet-connected security gateway can offer policy-configuration services for a fee. For future work, we plan to present techniques to construct a secure realization of this architecture and analyze possible deployments in the contexts of some widely publicized attacks.

%We summarize the main contributions we make in this paper as follows:
%\begin{itemize}
%	\item We propose a \remoteGate architecture, in which any Internet-connected security gateway can offer policy-configuration services for a fee;
%	\item We present techniques to construct a secure realization of this architecture;
%	\item We analyze possible deployments in the contexts of two recent, widely publicized attacks.
%\end{itemize}

\begin{figure*}
	\centering
	\includegraphics[width=35em]{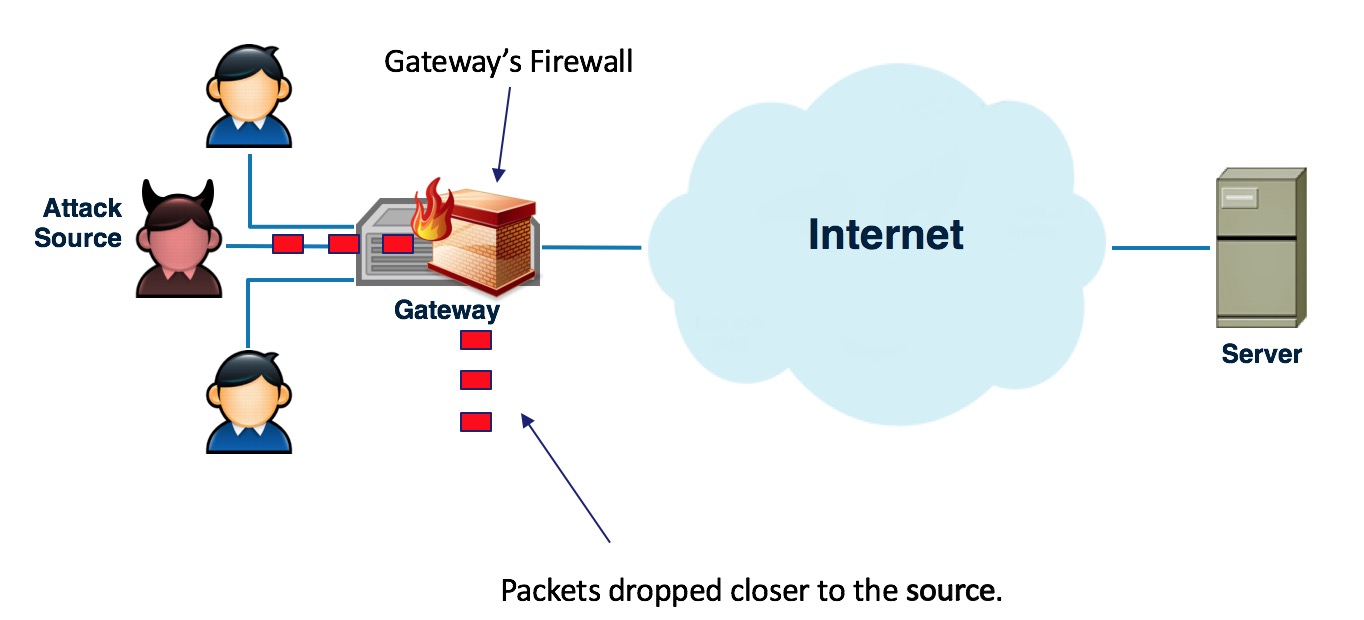}
	\caption{In our remote-configuration scenario, the server wishes to have the gateway block malicious traffic coming from attack source.}
	\label{fig:model}
\end{figure*}

\subsection{Motivation}

The greatest challenge is designing such a system for \remoteGate is understanding the benefit--risk tradeoffs seen be each party involved. We describe in this section this tradeoff for each type of party. For now we assume that there are three parties (see~\autoref{fig:model}): the attack source, the victim of the attack (server), and the security gateway (a system that can filter traffic between the attacker and the server). The main mechanism considered is for the server to send a request to the security gateway, and for the security gateway to change its security policy to block attacks originated from the attacker and destined for the server. 

\paragraph*{Motivation for the server}
The main goal for the server is to improve its network security, both in terms of accuracy and cost. In terms of accuracy, remotely configuring the policy of a security gateway allows the server to achieve a security policy that is customized for the attacker's network. This implies that the security policy resulting from the remote configuration takes into account both the attack traffic that is afflicting the server and the normal traffic that passes through the security gateway. As a simple example, denial-of-service attack traffic is best blocked by filtering it as close to the attacker as possible, as filtering it later in the network becomes almost impossible due to bandwidth constraints and due to similarity to normal traffic.

In terms of cost, remote policy configuration effectively offloads some of the security enforcement typically done at the server's own firewall to the security gateway. This is particularly useful for servers who run popular Internet-facing services and would need extremely powerful firewalls to handle all of the incoming traffic at a single ingress point. As a matter of fact, highly scalable Internet services often operated in a distributed fashion from multiple data centers, complicating the security management of all the firewalls involved. Offloading the security policy to remote security gateways would allow for cheaper and easier management of firewalls closer to the server, while making use of the large number of security gateways sitting mostly idle in front of the users accessing those servers.  

Underlying these goals of improved accuracy and lowered cost is the key need to maintain the same security guarantees as in the local firewall scenario. Because the enforcement of the security policy is now shifted to a remote security gateway not controlled by the server, our system needs to explicitly account for the potential loss of trust.

\paragraph*{Motivation for the Security Gateway}
The security gateway operates under one requirement, to allow all desired traffic sent by the network it protects to reach the Internet. In general, the definition of \emph{desired traffic} is not known and thus the gateway mostly operates by allowing traffic that it has seen before or that was explicitly whitelisted. In this context, the request to change the security policy sent by the server (which is an arbitrary node on the Internet from the point of view of the security gateway) introduces two problems. First, the gateway has to account for the additional resources to be spent configuring and deploying a new security policy. Second, the gateway has to ensure that the new policy still allows all desired traffic to reach the Internet.

To motivate the security gateway, we propose that it is paid for the resources it spends (e.g., additional computation time, additional memory, additional network latency) to handle the security policy requested by the server. To ensure that the new policy still satisfies the existing needs of the security gateway's protected network, our approach will use the existing security policy as a starting point.

\paragraph*{Deriving Security Policies from Examples}

An significant, operational roadblock for the security gateway is in handling the format and semantics of the policy-change request. Different gateways have completely distinct filtering semantics, distinct languages for specifying the filtering policies, and distinct levels of expressiveness for capturing a certain filtering functions. Exposing this information to anyone on the Internet, beyond the challenges of actually representing this information in a common format that everyone agrees upon, is also a security risk for the security gateway, as attackers can study the filtering languages, its expressiveness, and its semantics to find weaknesses that would allow evasion attacks.

In our approach, we decouple the policy engine of the security gateway from the server's request by having the server provide examples of attacks (i.e., network packets that form an attack) for the security gateway to analyze and create a security policy from. This restricts the communication between \remoteGate and the gateway to exchanging information about the inputs to the policy and the expected results ("block" or "allow"), without requiring a particular policy type or a particular enforcement engine design. The security gateway can use a simple firewall-style, rule-based system that only takes into account IP addresses and ports, or it can use a machine learning classifier that uses the whole contents of a network packet to make decisions.

\subsection{Model}
Our model consists of a \emph{server} $S$ that can receive packets from \emph{clients} over the Internet. Specifically, consider a network of clients connected to the Internet via a network \emph{gateway} $G$. We are interested in the communication of these clients with $S$. We assume there are multiple (at least two\footnote{The server's firewall and a gateway closer to the clients.}) gateways on the way from each client to the server. Each gateway is equipped with a \emph{firewall} that can control the flow of traffic sent to the Internet via a \emph{security policy}, which basically identifies specific characteristics about data packets passing through the gateway. As mentioned previously, usually these policies filter incoming packets to the local network behind the gateway. However, \remoteGate will focus on filtering the outgoing packets instead.

We consider an \emph{adversary} who controls some or all of the clients behind the gateway to attack the server by sending attack packets. We assume that this adversary does not collude with the gateway.\footnote{In Section~\ref{sec:discovery}, we present a gateway discovery algorithm with spoof checking that will help ensure that there is no such collusion. However, we leave it for future work to find such a gateway (that controls the traffic between the attacker and the server without any collusion with the former) with higher efficiency and precision.} The server wants to remotely configure the gateway's policy to impede the attack traffic targeted at itself by these adversarially controlled clients. 

We model the server and the gateway as rational players by provisioning incentives for both to participate in this protocol through carefully designed utility functions. In particular, we incentivize the gateway to participate in this protocol by providing a payment from the server upon successful deployment (and maintenance for a specified period of time) of a rule that prevents the attack. This can be done via any secure (digital) payment mechanism, that will allow both the server and gateway to contest the payment in a fair manner any time there is a conflict in the future. Such a mechanism is independent of our protocol and can be achieved by using a payment processing networks such as Visa~\cite{visa} or any such equivalent (even digital currency like Bitcoin~\cite{bitcoin} etc.).

\subsection{Technical Challenges}
Several technical challenges arise when dealing with a problem at the scale of the Internet, specifically given the anonymity of the remote users and their unknown intents. In this paper, we attempt to tackle two challenges that we think are most important to our problem statement and our model. First, we try to answer the question of how one can design a reward model that is incentive-compatible and fair to both the server and the gateway at the same time. Our solution achieves this by making the server pay the gateway a reward that is proportional to the accuracy of the rule that he deploys in its firewall. This incentivizes the gateway to participate in the protocol and more importantly, our design of the reward mechanism forces the gateway's dominant strategy to play honestly. Similarly, we model the server's incentive by exploiting its need to get the attack prevented as soon as possible. This allows us to force the server's dominant strategy to be incentive-compatible as well.

Another important challenge we tackle is for the server is to ensure that the gateway deploys the correct model and continues filtering packets based on it. We refer to this problem as the \emph{deployment verification} problem. This is important because it requires the server to not regret paying the gateway later upon realizing that the gateway maliciously deployed the wrong model (or no model at all). However, this is tricky because we never require the server to explicitly ask the gateway for any information about its firewall or the rule generated (apart from what is leaked from the verification during rule learning), for obvious privacy reasons. In this paper we tackle this problem by making the server specify the period for which it requires the gateway's service and pay a periodic fee for that period. The idea is that this fee will incentivize the gateway to retain the model in its firewall. Further, any breach, if encountered, is assumed to be resolved by a trusted third party, so later if there is still attack going on, the server can take help from this third party to resolve the conflict. This assumption of a trusted third party seems to weaken our model at first, but we emphasize in the paper that this is certainly not the only solution to this problem.

\subsection{Our Contributions}
This paper introduces \remoteGate, an efficient protocol for remote configuration of security gateways. The key contributions of \remoteGate are as follows:
\begin{enumerate}
	\item An incentive compatible mechanism for a server under attack to convince a remote gateway to change its firewall settings in order to prevent the attack.
	\item A light-weight interactive mechanism for some party (server) to help a remote party (gateway) learn a classifier for the data that the former provides, without revealing any more information than the accuracy of the classifier on the specified test set.
	\item A spoof checking mechanism to verify the source address mentioned in the attack packets. 
\end{enumerate}
Apart from these main contributions, our vision in designing \remoteGate is to provide a protocol that can be autonomously run by devices that are connected to the Internet. Although we restrict the discussion in the paper to a server and a gateway, the protocol as designed can be used with any device capable of sending and receiving messages over the Internet. With the advancing technology in hand-held devices, personal computers and the advent of pervasive computing and Internet of Things (IoT), \remoteGate can be used by any smart device to protect itself against incoming attack packets. 

Moreover, several other aspects of \remoteGate are voluntarily kept open to adapt to the emerging technologies. Be it the learning algorithm used by the gateway to obtain the rule to be deployed, or the kind of attack packets against which the defense is sought, \remoteGate allows for a wide variety of security services to be provided, not only on a local scale but on a global scale. Even with the digital payment system that will be used for reward payment or the service fee payment, upcoming trends like blockchain and cryptocurrencies can play a vital role in providing the required security guarantees that we need. We emphasize that the \remoteGate aims at fostering the ideology that while firewalls provide safety against incoming traffic, they can also be efficiently used to stop the attack packets from entering the Internet in the first place. 

\subsubsection{Properties of \remoteGate}
With the objectives as above and the vision of global software defined networking in mind, our design of \remoteGate has the following properties. 
\begin{enumerate}
	\item \textbf{Filtering Accuracy: }The accuracy of the rule deployed at the gateway's firewall is highly dependent on the way that rule is derived. \remoteGate aims to achieve a higher filtering accuracy by outsourcing the learning to the gateway and not generating the rule itself. This helps the gateway take into consideration the existing rules in its firewall to learn an optimal model that uses the limited computational resources judiciously and benefits from the rules for other outgoing traffic from the gateway. Moreover, this approach prevents the server from accidentally (or not) blocking other outgoing traffic through its firewall.
	\item \textbf{Decentralized Attack Prevention: }A common practice for attack prevention these days is to inform companies like Cisco, Netgear, Fortinet etc. about an ongoing attack upon which they observe the traffic and design rules or patches to be installed in the existing routers/firewalls. This service is often provided in exchange for a fee, which is decided by the service providers themselves. All seems good until we realize that this approach relies on a critical assumption that the attack packets have not originated out of collusion with the prevention mechanisms. \remoteGate envisions to remove this assumption by directly contacting the gateway to learn a rule and deploy in its firewall without involving any third party. This removal of trust in external entities helps protect certain types of attacks that are designed with a financial motive in mind.
	\item \textbf{Incentive-Compatible Reward Mechanism: }\remoteGate uses concepts from game theoretic mechanism design, specially principles from bilateral trading to devise negotiation schemes and auction mechanism for the reward that is agreed upon by both the server and the gateway. Since the two parties involved are modeled as being rational and mutually untrusted, this design strategy guarantees that the reward chosen is fair with respect to the true valuations of both the server and the gateway. Acting honestly is shown to be the dominant strategy in the \remoteGate protocol.
	\item \textbf{Light-Weight Protection: }By offloading filtering of attack packets as well as learning of these filters to the gateway, the server now runs a light weight \remoteGate client that is able to carry out our protocol without any need of a sophisticated hardware or infrastructure in general.  
	\item \textbf{Monetary Fairness: }An obvious question that comes to mind when thinking about \remoteGate's approach for attack prevention is how one can be sure that the gateway will eventually deploy the rule in its firewall and not just abort after the payment has gone through. If one tries to address this issue by paying the gateway at the very end of its services, then the same question can be asked on the server's end where there is no guarantee of the server aborting the protocol just before the payment begins. \remoteGate deals with this issue through (1) a periodic payment scheme designed to reward the gateway for its services so far, and (2) a conflict resolution protocol, which can either involve a trusted third party or use a decentralized solution like Ethereum smart contracts to ensure that any breach of agreements is caught and dealt with fairly. This way, \remoteGate ensures that both the server and the gateway are unable to put themselves into a monetary advantage by simply aborting the protocol at the right time.
	\item \textbf{Smart Protection: }We have designed \remoteGate trying to keep it as automated as possible. As long as the protocol is able to differentiate the attack packets from good packets and set some initial paramters, that is all we assume about what it needs to initiate the communication. An interesting question that comes to mind here is how does the server verify the identity or source of the attack packets, given that it is relatively easy to spoof address on the Internet. We handle this using a spoof check subroutine built into the \remoteGate protocol that first tries to find a gateway to interact with and only when it is convinced that the gateway is routing traffic from the attacker.
	\item \textbf{Small number of firewall installations: }Consider a situation in which thousands of computers around the world are under a distributed DoS attack originating from somewhere behind, say, the Tor network. If each of these computers were to install an expensive update to their respective firewalls that can curb the attack, the combined cost of this installation would far exceed the cost incurred in deploying a rule in a gateway (or two) that are closest to the source. The idea here is that once a rule is deployed in a gateway, it can prevent the attack packets from being routed to any victim of the attack over the Internet. 
	\item \textbf{Wide Applicability: }The vision of \remoteGate is to be a light weight client that can be installed widely in any device, ranging from heavy duty routers and servers to handy hand-help devices like the cell phone and other forms of IoT compatible devices. This way, attack packets are stopped for devices that are not firewall-compatible or even those that are not sophisticated enough to run antivirus functionality. 
\end{enumerate}

\subsubsection{Limitations of \remoteGate}
\label{sec:limitations}
There are several limitations of \remoteGate and areas of improvement that future research can look into. The solution provided in the upcoming sections are far from optimal, keeping in mind the novelty of attacks that exist over the Internet everyday. Nonetheless, we hope to achieve security in the situations where our assumptions hold and mention some important limitations of our protocol, as we see it, below.

\begin{enumerate}
	\item \textbf{Slow Responsiveness under DoS Attack: }\remoteGate, with its current algorithm, can be less responsive when the server is under a severe DoS attack. More concretely, if the attack packet frequency is too high, the server may not be able to communicate with the gateway at all in order to carry out the required message exchange. In this situation, the server should resort to conventional methods to stop the attack.
	\item \textbf{Dependence on Learning Algorithm: }An important assumption that \remoteGate makes is that the gateway is using a learning algorithm that will be able to classify the packets as required. For this reason, \remoteGate does not explicitly endorse the use of any specific learning technique for the gateway to use. Rather, it allows the gateway to choose one by itself, keeping in mind its rationality with respect to the the reward model. Of course, a (set of) particular algorithm(s) can always be hard-coded into the system if the situation demands so.
	\item \textbf{Classification of Packets: }One might challenge the classification provided by the server to the gateway on the training and test examples for the attack packets for its accuracy. In other words, does the server has an incentive in lying on the classification of the packets? If no, how does it obtain the information about these packets in the first place? For the first question, we assume that when trying to stop the attack, any false information provided by the server only delays the time it takes for the gateway to produce the right model. Thus, if the server chooses to lie on the classifications in some round, only to report the correct classification later, the gateway will be able to identify this behavior in the later stages of the algorithm and abort the protocol. This takes away the incentive from the server to lie and hence, play honestly. For the second question, we assume that the server is informed about the classification either through some threshold scheme (based on the frequency of packet arrival) or a human-assisted mechanism. Ofcourse, this question aligns with the broader concern about not performing the filtering at the server's end. Our answer to this analogy is that (1) the server needs to know \emph{what} it doesn't want to receive, and (2) the gateway is a bette candidate for rule deployment for certain types of attacks, as explained next.
	\item \textbf{Single Point Protection vs. \remoteGate: }Is it better to configure the server's firewall or hunt for an external gateway? The answer depends on the type of attack faced by the server. If the server receives attack packets from multiple sources but containing similar content, then an easier, and probably economic, solution would be to deploy a rule in its own filtering service and stop the packets. However, if the source is suspected to be local to a small group of users, then \remoteGate offers a promising solution, both in terms of preventing the packets from entering the Internet in the first place and possibly obtaining a higher filtering accuracy (as explained previously). The exact decision boundary, however, is, to some extent, subjective and thus, open to exploration and further research.
	\item \textbf{Requirement for Good Packets: }Any model generation algorithm that will filter attack packets from good packets will require samples of both types to create a sufficiently accurate boundary that does not suffer from overfitting or under-fitting. Assuming that the attack packets are known, it is entirely possible that the set of good packets is much wider and hence, the choice of an optimal set of good packets to use for training becomes crucial. In other words, if the attack packets all contain a pattern string in them, then every packet that does not contain that pattern string is a good packet. The important question here is how to represent the packets in a way that does not make the server send too many packets to the gateway before a correct model is built. We acknowledge that \remoteGate leaves this problem open for future work and for now, assumes that the server contains an optimal set of example packets for the gateway to train on.
	\item \textbf{One Model per Server: }For every server that contacts a given gateway for an attack, does that gateway deploy a different rule for each of these requests? If yes, then this is an indirect DoS attack on the gateway itself and its firewall, which will likely not be able to scale up to so many requests. However, if not, then how does the gateway decide the number of rules to deploy in order to maximize the filtering accuracy as well as the reward collected? This is one of the reasons why the choice of learning algorithm is kept open for the gateway to decide, since it can choose to build rules that, instead of serving the request by only one server, cater to a set of servers together. This is equivalent to developing rules that perform multi-class classification over the space of network traffic, instead of the binary case of separating attack from non-attack. As mentioned earlier, this is open to future work for further exploration and optimization.
	\item \textbf{Collusion: }Another important assumption \remoteGate makes is that the gateway and the server are not colluding with the attacker. Of course, in reality, this collusion is a very easily conceivable possibility, which must be dealt with in extensions of \remoteGate to malicious server/gateway settings. However, we emphasize that the current assumptions about the rationality of the two parties involved still make sense since there must exist at least one (non-colluding) gateway between the attacker and the victim to have any hope of attack prevention at all. The challenge then lies in the discovery of this \emph{good} gateway and carry out the protocol with it. Our algorithm for gateway discovery hopes to achieve this, although, its current version may report false positives in certain scenarios. 
\end{enumerate}

\subsection{Paper Organization}
This paper is organized into several sections. Section~\ref{sec:relatedWork} discusses the related work in this area. 
%Section~\ref{sec:background} provides details of some technical concepts that are used to design the protocol. 
Section~\ref{sec:remoteGate} presents the \remoteGate protocol and provides details about its various steps. 
%Section~\ref{sec:evaluation} provides some empirical observations and early projections about the performance of our protocol. 
Section~\ref{sec:discussion} presents discussion on some technical aspects of the protocol and sheds more light on the decisions made during specific stages of the protocol design. Section~\ref{sec:future} discusses some interesting open problems and opportunities for future work related to improvement of \remoteGate as well as the general idea of remote configurability. Finally, section~\ref{sec:conclusion} concludes our paper and highlights our main results and ideas. The reader is also encouraged to refer to the Appendix for some interesting FAQs (frequently asked questions) that we encountered during the design of \remoteGate and our answers to those questions.

\section{Related Work}
\label{sec:relatedWork}
%\paragraph{Remote configurations/router management}
%\todo[inline]{Complete this}
%
%\paragraph{Machine learning based firewall management}
%\todo[inline]{Complete this}

The main task in our protocol is for the gateway to learn a rule that correctly classifies the attack packets from the safe packets, using information provided by the server. As mentioned before, this can either be done using some static approach like identifying the IP address or information about the ports, or a more dynamic approach of machine learning that is applicable to even the most sophisticated attack packets. We use of the latter in our paper to make \textsc{RemoteGate} applicable to a wider variety of applications. 

Machine learning is emerging as a popular area for both the research community and the service industry. Companies like Amazon~\cite{amazonML} and Microsoft~\cite{team2016azureml} are already in the game along with academic research groups~\cite{potluru2014cometcloudcare} towards providing users around the world with online machine learning visualizations and training infrastructure, while promising privacy of user data. However, this emerging trend of machine learning as a service~\cite{ribeiro2015mlaas} comes with its own challenges. One of the first things that come to mind is incentivizing people for providing computational assistance with the training. Numeraire~\cite{numeraire} uses an auction mechanism to reward participants with the Numeraire token, while any mechanism for digital payment can be used as a financial incentive. 

The bigger picture, however, does not arise from financial incentives. Malicious intentions and privacy breaches pose problems to the use of any service that aims to be publicly available and relies on user data. This source of disruption can either be from the side of provider of data points, where the intention is to bias the output of the classifier towards incorrect classifications, or it can be from the side of the model generator where the client does not trust the service provider with respect to its services. For the former, techniques that handle adversarially generated examples and aim to create robust models~\cite{kurakin2016adversarial, goodfellow2014explaining, wang2016theoretical, dalvi2004adversarial, laskov2010machine, zhou2012adversarial} come in handy. While most of these works focus on recovery of the classifier after an adversary injects infected examples, there is also research like in~\cite{lowd2005adversarial} that tries to model this problem of constructing adversarial attacks by learning sufficient information about the classifier. The idea here is to better understand and determine the loopholes and the weaknesses in the learning algorithm to make them more robust to clever attacks in the future.

The scope of this paper, however, is along the lines of an untrusted service provider (the security gateway in this case). One particular example in this setting is performing execution of neural networks on an untrusted cloud, where the client who requires this computation is not sure if the cloud performed the inference correctly or not. An approach here is to require some form of proof from the cloud that convinces the client of the computation performed. There are many ways that this proof can work. SafetyNets~\cite{ghodsi2017safetynets} take an approach that is very specific to the underlying classifier (neural network in this case) by exploiting some mathematical properties of the learning algorithm. However, to generalize the proof technique, several other options come in handy. Some of the popular ones are listed as follows.
\paragraph{Verifiable Computation} \cite{wahby2016verifiable, wahby2015efficient, braun2013verifying, parno2013pinocchio} This approach aims at answering the fundamental question of how a local computer can verify some computation that it asks a remote server to perform, without explicitly performing the computation itself? There are may ways people have attacked this problem, ranging from expressing the computation in some high level language and then using the features of that language to generate a verifiable proof, to exploiting some mathematical properties of the underlying function to generate the proof. The idea is for the verifier to perform minimal computation in order to validate the prover's claims. In our setting, we perform this verification for the model generated by the gateway interactively during the training phase itself, without the server having to explicitly run a learning algorithm locally. Since we are doing this over a resource limited network, we perform verification keeping in mind both the computational resources at the server and the limitations on latency as posed by the network and the severity of the attack.
\paragraph{Interactive Proofs} \cite{goldwasser2008delegating, reingold2016constant, goldreich2017simple, thaler2013time, cormode2012practical,ghodsi2017safetynets} Yet another frequently used approach for verifying computations is to use an interactive technique in which the prover and the verifier participate in a sequence of message exchanges that allows the verifier to reach a verdict on the prover's claims. Often these exchanges happen in a way that the verifier queries the prover at certain checkpoints of the computation, and if it receives convincing responses, it delegates that the prover is performing the computation as required. The property that is crucial to ensure here is that if the prover performs the computation correctly, then it can always convince the verifier of the same, but if not, then we can bound the probability of the verifier being fooled by the prover by any desired value. Often this comes with a trade off in the number of bits exchanged as the error tolerance goes down. In our approach, as mentioned previously, we use an interactive scheme as well, in which the server and the gateway exchange messages until the former is convinced that the latter has generated the correct model for classifying the attack packets. Additionally, we use this interactive scheme in a clever way by using the progress of the gateway as a metric for the server to decide what reward to offer in the future rounds. This incentivizes the gateway to strive for best-effort model generation as the rounds progress and collect a high reward at the end.
\paragraph{Zero-Knowledge Proofs} \cite{goldreich1994definitions, chiesa2015cluster, duan2008practical, pathak2012privacy, yampolskiy2011ai} In many applications, privacy of the model learnt or the underlying data is of utmost concern when two or more parties are involved. Consider a setting where a server knows some private database and a client wants to perform a query over the database, whose result must not reveal any more information about the database than what is revealed by the result itself. Zero-knowledge proofs are cryptographic tools that are able to offer this security in a succinct and non-interactive manner. For our application, we could have used this technique since the gateway's firewall settings are its private bits, but \textsc{RemoteGate} offers a much simpler solution in which all the server needs is to be convinced that the gateway has the right model without requiring explicit knowledge of the model itself. In fact, the server will never query the gateway in the future about this model once it has been correctly deployed and is also willing to offer some tolerance to the accuracy obtained. All these relaxations allow for a simpler solution using an interactive scheme rather than using something as sophisticated as a zero-knowledge proof, for which we would have had to make assumptions regarding the gateway being powerful enough to generate such a proof and the server being equipped with zkSNARK(s) or equivalent to be able to verify such a proof. 
\paragraph{Secure MPC} \cite{mohassel2017secureml, lindell2009secure, ohrimenko2016oblivious, lindell2000privacy, chaudhuri2009privacy, ccatak2015secure} The problem in this paper can also be studied under the general umbrella of secure multi-party computation, in which a computation task is performed between multiple parties in a way that satisfies some security properties specific to the application. One way we could have moulded our problem in this framework is to have allowed the server and the gateway to collectively compute the classifier for the attack packets and then use this model to deploy at the gateway's firewall. However, MPC is usually a time intensive process and with the leaning problem being potentially arbitrarily complex, it can take much longer time than our simple interactive scheme. We emphasize that during attack prevention, the server may only have limited time before it can no longer communicate with the gateway, so we need a solution that runs in as small a time and using as little computation as possible.
\paragraph{Probabilistically Checkable Proofs} \cite{setty2012making, setty2012taking, setty2013resolving, ishai2007efficient} There is yet another technique that is used for proof generation in which the prover computes a bit string which the verifier can query for a small number of certain fixed locations and only verify that the bits on those positions are correctly computed. This provides probabilistic guarantees on the computation task performed since the verifier gets to choose the bit locations it wants to query. However, these strings are very long (often exponential in the input length) and hence requires both the verifier and the prover to have rich computational resources to participate in such a proof system. Again, in our setting, this would result in a heavy weight solution which would constrain the resources available at the server's and the gateway's end.

%%%%%%REMOVE THIS%%%%%%%%%%
\paragraph{Incentive Compatibility.} In order to make our protocol force the gateway to act honest, we use the popular approach of Incentive compatibility~\cite{nisan2007algorithmic} from game theory. Incentive compatibility is a property by which participants in a mechanism achieve best outcomes to themselves by acting according to their true preferences. We use a special form of this compatibility, called the \emph{dominant strategy incentive compatibility (DSIC)}, in which we guarantee that following the true beliefs in presenting preferences is weakly dominant in the sense that no deviation from this strategy will increase the payoff (reward in this case). We claim that the DSIC strategy for the gateway, with respect to the algorithm we provide for the server, is to act honestly. We achieve this by using the ideas from some popular number guessing games~\cite{nisan2007algorithmic, mendes2014guessing}. 

\section{The \remoteGate protocol}
\label{sec:remoteGate}
In this section, we present a general overview of the \remoteGate protocol to enable a server to contact a remote gateway and help it deploy a rule in its firewall that will filter out the attack packets from reaching the server. We provide details of the individual step of the algorithm subsequently.

\subsection{Protocol Overview}

\begin{figure*}[h!]
\captionsetup{justification=centering}
	\includegraphics[width=45em]{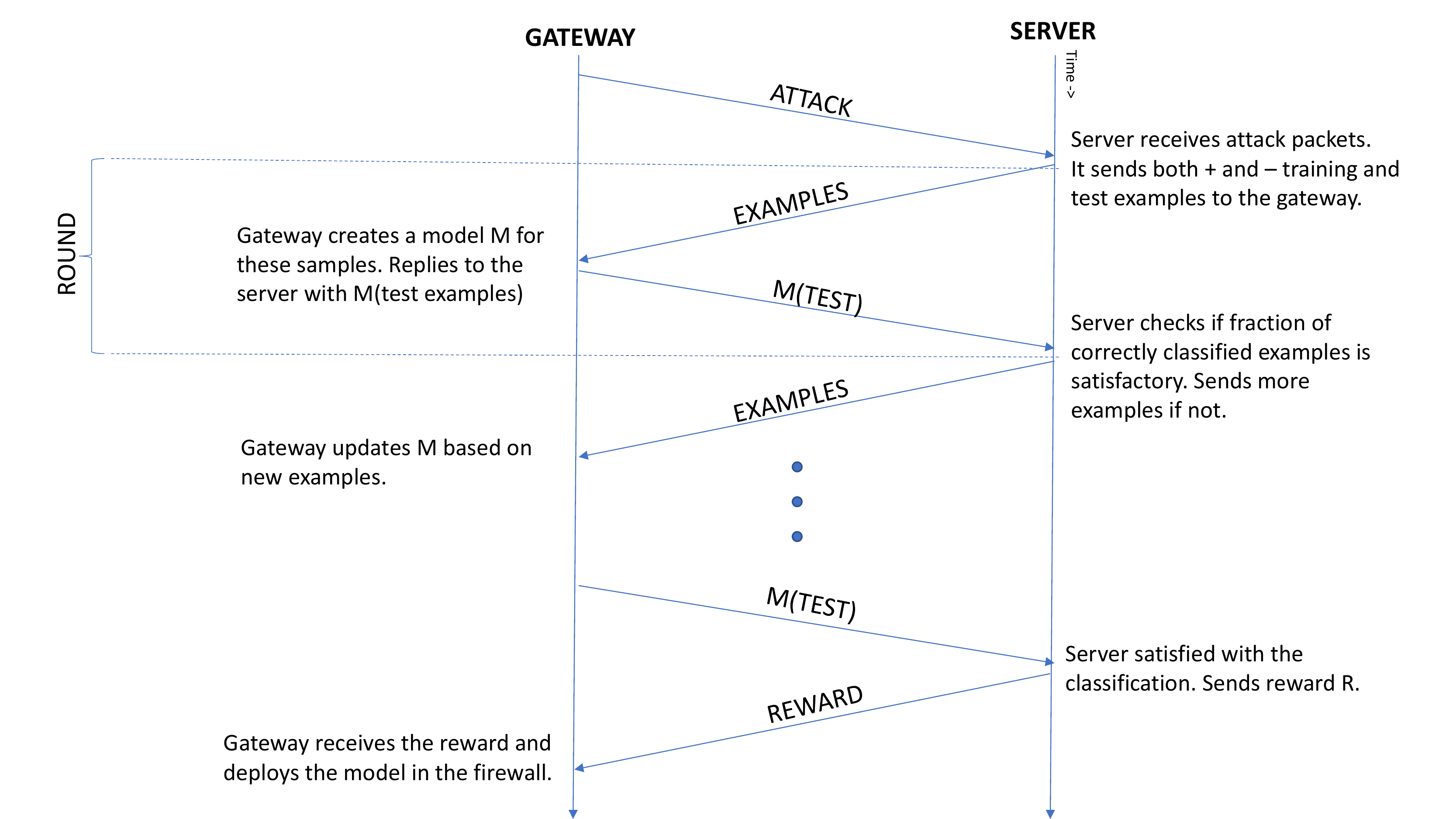}
	\caption{Schematic of the communication between the server and the gateway during a typical run of \textsc{RemoteGate} protocol.}
	\label{fig:schematic}
\end{figure*}

Our protocol consists of three main steps : \emph{Discovery, Learning} and \emph{Service}. In the \emph{Discovery} phase, the server uses the information in the attack packets to discover a list of gateways that are forwarding the traffic from the attack source. Upon obtaining this list, the protocol enters the \emph{Learning} phase, where the server initiates communication with the farthest gateway\footnote{Based on the information on round trip time (or some other metric) provided by the discovery algorithm. See Algorithm~\ref{alg:discovery} for our approach.} in this list and helps it learn the required model (policy) that will filter the attack packets. Once this learning is complete and the server pays the gateway for the deployment of this model, the \emph{Service} phase begins in which the server pays the gateway a(n agreed upon) fee, periodically, to retain the model in the gateway's firewall. This resembles an establishment of a service-agreement between the server and the gateway, which in case of a dispute, can be resolved by an external trusted third party. 

A schematic in Fig.~\ref{fig:schematic} provides a high level view of the major steps of our protocol. As mentioned before, the server, upon receiving attack packets, initiates an exchange with the gateway by providing the latter with positive and negative examples of the attack packets that the server has received. The gateway uses these examples to create a model that produces a classifier based on these training examples and provides a proof of this construction by classifying the test examples that the server provided. This way the gateway reveals only as much as much information about the model as is obtainable by the classifications it just provided. Upon receiving these classifications on the test examples, the server checks if the error in classification is within a pre-decided tolerance limit. If satisfied, the server pays the gateway the promised reward and the gateway deploys the model in its firewall for the packet filtering to begin. Otherwise, the server issues another round of training and testing with the gateway, on fresh examples, until it receives classifications that are \emph{sufficiently} accurate. 

The \remoteGate timeline is depicted in Fig.~\ref{fig:remoteGateTimeline}, in which the various events at different points of time have been specified. At $T_{start}$, the server starts receiving attack packets from the attacker, but has yet not identified that it is under attack. After receiving these packets for some time, the server detects the attack at time $T_{detect}$ and starts collecting samples for \remoteGate learning protocol. After sufficient number of samples have been collected, the server begins the \remoteGate protocol at time $T_{begin}$. It starts by discovering the list of gateways to interact with and once. Once the candidate gateways have been discovered at time $T_{discovery}$, the server runs a spoof check with these gateways to determine what gateway to interact with. At time $T_{spoof}$, this gateway is determined and the server starts negotiating for the initial reward. Upon successful negotiation, the learning phase begins at time $T_{init}$ and the server engages in rounds of training and testing with the gateway. After the $r$ rounds of learning have completed (which satisfy the server's error tolerance) at time $T_{learned}$, the server pays the gateway the promised reward. At time $T_{paid}$, this payment is successful and the gateway prepares deploying the model in its firewall. Once the model is successfully deployed by time $T_{deployed}$, the \remoteGate service period starts during which the filtering of attack packets happens through the gateway's updated firewall. This service ends at time $T_{end}$ after which both the server and the gateway terminate the protocol.   

We provide details of each of these steps in the upcoming sub-sections. We would also like to refer the reader to  Appendix~\ref{app:technical} in this paper for some frequently asked questions that we encountered with respect to our approach in this paper.

\begin{algorithm*}[h!]
\caption{\remoteGate protocol for the server.}
\label{alg:remoteGateServer}
\begin{algorithmic}[1]
\Procedure{\remoteGate-Server}{$ \mathcal{A}, \Delta$}
	\State $\texttt{GList} \gets \Call{GatewayDiscovery}{\mathcal{A}}$. \Comment{\texttt{GList} ordered by farthest gateway first.}
	\State $\texttt{FList} \gets \{ G \in \texttt{GList} \mid \Call{Server-SpoofCheck}{G,\mathcal{A}, \Delta} = \texttt{false}\}$. \Comment{Remove spoofed gateways.}
	\For{$G \in \texttt{FList}$} 
		\State $\texttt{LOG}_S \gets \emptyset$ \Comment{Initialize server log.}
		\State Initialize $\varepsilon, v_S$ and $\Gamma_\text{service}^S$.
		\If{$\Call{Server-StartService}{\varepsilon, \mathcal{A}, v_S, \Gamma_\text{service}^S} = \texttt{false}$}
			\If{$\texttt{LOG}_S$ contains "\texttt{BREACH}"}
				\State Issue a call to \Call{Conflict-Resolution}{$\texttt{LOG}_S$}.
	\EndIf	
	\Else
		\State \textbf{Output} $\texttt{LOG}_S$ and \Return
	\EndIf
	
	\EndFor
\EndProcedure
\end{algorithmic}
\end{algorithm*}

\begin{algorithm*}[h!]
\caption{\remoteGate protocol for the gateway.}
\label{alg:remoteGateGateway}
\begin{algorithmic}[1]
	\Procedure{\remoteGate-Gateway}{$\ $}
		\While{\texttt{true}}
		\If{\texttt{PING} packet received from server $S$ with attack packet $A$ and $\Delta$}
			\If{$\Call{Gateway-SpoofCheck}{S,A, \Delta} = \texttt{false}$}
				\State Wait for the \texttt{INIT} packet from the server. \textbf{break} if timed out.
				\State $\texttt{LOG}_G \gets \emptyset$ \Comment{Initialize gateway's log.}
				\State $\Call{Gateway-StartService}{\texttt{COMMIT}(\rho_{S,1}), \Gamma_\text{service}^S, A}$
				\If{$\texttt{LOG}_G$ contains "\texttt{BREACH}"}
					\State Issue a call to \Call{Conflict-Resolution}{$\texttt{LOG}_S$}.
				\EndIf
				\State \textbf{Output} $\texttt{LOG}_G$.
			\EndIf
		\EndIf
		\EndWhile
	\EndProcedure
\end{algorithmic}
\end{algorithm*}

\subsection{Initialization}

The \remoteGate system is installed as an application on both the server and the gateway for them to engage in the sequence of communications as required by our protocol. We provide a high level description of this application in Algorithms~\ref{alg:remoteGateServer} and~\ref{alg:remoteGateGateway} (for the server and the gateway, respectively). 

The server starts \remoteGate by determining the average time interval $\Delta$ between two attack packets in its set $\mathcal{A}$ of all the attack packets that have been received so far. Once this information is obtained, the server launches a discovery protocol (Algorithm~\ref{alg:discovery}) to obtain a list, denoted $\texttt{GList}$, of gateways that may be forwarding these attack packets. The gateways in this list are marked with some additional information about how far the gateway is from the server (for example, by using round trip times). However, this list is prone to containing gateways whose identities may have been spoofed or who were wrongly discovered. The server removes such \emph{false-positives} by running a spoof check (Algorithm~\ref{alg:serverSpoof}) with each of the gateways in the list and obtains a filtered list $\texttt{FList}$ of gateways, ordered similar to $\texttt{GList}$. Starting from the first gateway in this list, the server begins with resetting the protocol communication log, denoted $\texttt{LOG}_S$, and based on the gateway, decides on what error tolerance $\varepsilon$, value of the reward $v_S$ and duration of service $\Gamma_\text{service}^S$ to negotiate. With these parameters, the server then begins the \remoteGate protocol with the gateway. If the protocol was successfully deployed and serviced for the negotiated time period, the server terminates the protocol with success. However, if at any time during the protocol the server's log contains the string "\texttt{BREACH}" (due to detection of malicious behavior from the gateway), the server resorts to the third party for appropriate conflict resolution. Else, it restarts this process with other gateways in $\texttt{FList}$. 

\begin{figure*}[t]
	\captionsetup{justification=centering}
	\includegraphics[width=\textwidth]{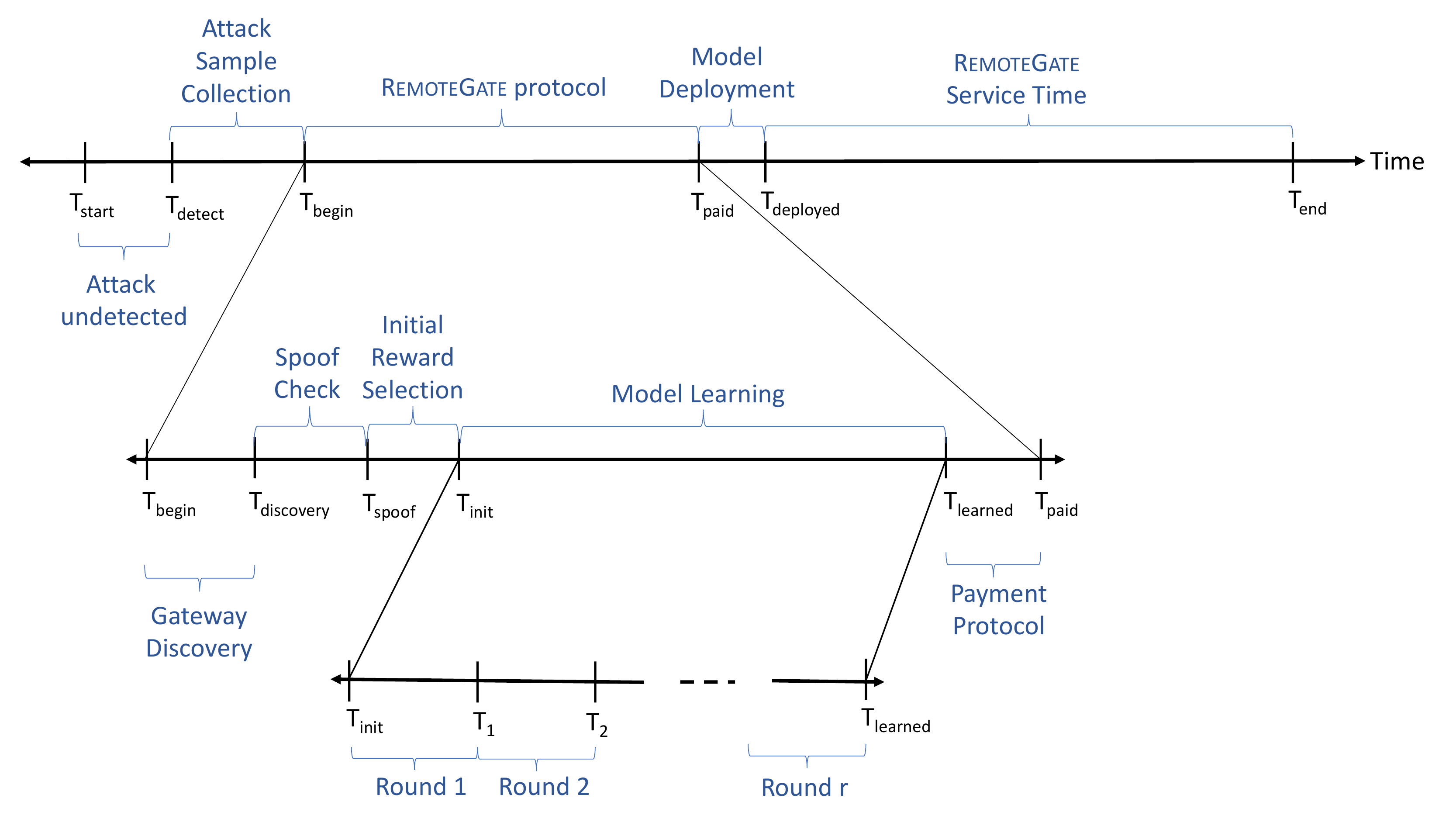}
	\caption{\remoteGate timeline diagram showing what events happen at various times during a run of the protocol. The plot is not drawn to scale. With respect to this plot, the total time taken by \remoteGate to prevent the attack is $T_{deployed} - T_{detect}$.}
	\label{fig:remoteGateTimeline}
\end{figure*}

Unlike the server, the application installed at the gateway's end (Algorithm~\ref{alg:remoteGateGateway}) pools until some server contacts it with a spoof check packet, called \texttt{PING}\footnote{Note that this ping packet is different from the standard ping packet in network routing protocols.}. At this point, the gateway recognizes that the server is trying to verify if the attack packet $A$ (with inter-packet time of $\Delta$) actually came from the local network behind the gateway. The gateway engages in this spoof check to ensure that its address was not spoofed by the attacker (Algorithm~\ref{alg:gatewaySpoof}). If the spoof check returns with the indication that the attacker belongs to the local network, the gateway begins the learning phase of the \remoteGate protocol with the server. Similar to the server's log, the gateway also maintains a log, denoted $\texttt{LOG}_G$, for all the communication it does with the server. If at any point of time during the protocol the gateway detects malicious behavior from the server, it reports it to the log as a "\texttt{BREACH}" string, and resorts to the conflict resolution protocol.  

\subsection{Gateway Discovery}
\label{sec:discovery}
%!TEX root = draft.tex

The first step for the server, once it is in posession of attack packets, is to discover the potential gateways that can protect the server by blocking any future attack packets. The list includes any gateway that is capable of blocking or filtering traffic from the attack source on its way to the server, and that is willing to participate in the \remoteGate protocol.

The challenge in enumerating the gateways on the path between two Internet nodes lies in the fact that there may be multiple paths between such nodes, each one used at different times or for different types of messages or for different directions of communication. The Internet is designed as a potentially directed communication graph, where routing policies can move traffic along unexpected routes given the observable topology. Commercial agreements driven by cost and traffic volumes, load balancing, privacy requirements, regulatory constraints, and availability and reliability needs ensure that routing policies at all levels (from local networks to large autonomous systems) are complex, often hard to infer, and uncertain to rely upon. We do not solve this problem here, rather we focus on the simple case of symmetric, stable paths, and then we discuss the corner cases left to solve in future work.

\paragraph{Symmetric, Stable Internet Paths}
Our approach is inspired by the \texttt{traceroute} diagnostic tool, which measures transit delays of packets between two nodes while also recording the latencies to each intermediate router on the path. Our Algorithm~\ref{alg:discovery} makes use of ICMP, just like \texttt{traceroute}, and could be perceived as an extension of ICMP messages in order to provide the additional information.\footnote{It is worth mentioning that some ICMP message types such as \textit{Router Solicitation} and \textit{Router Advertisement} appear relevant our gateway-discovery step, but we are not aware of ICMP messages or extensions that completely address our needs.} The basic mechanism is to send packets to the (apparent) attack source and to use time-to-live values to address the intermediate routers and request that they return their \textit{managed IP range}, if any. Let us denote such a request packet as \texttt{DISCOVER}. By ``managed IP range'' of a router we refer to the set of IPv4 addresses the router can reach on interfaces other than the one on which the \texttt{DISCOVER} packet was received, and for which the router is the sole entry--exit node. In other words, all traffic to and from the managed IP range passes through that router. While some routers on the Internet do manage IP ranges (as defined here), others do not, and thus we do not expect all gateways on a path between two Internet nodes to respond to a \texttt{DISCOVER} ICMP request. At a minimum, in typical configurations, there is a firewall on the server's node and a firewall on the attack-source's node that have their own managed IP ranges (respectively, the whole IPv4 address space except for the server's IP address, and the IPv4 address of the attack-source's node).

The server now holds a list of gateways, their correponding latencies, and their managed IP ranges. To proceed, the server sorts the list by latency, validates that each managed IP range contains the attack-source IP address, and selects the farthest gateway. We note that the managed IP ranges do not need to increase or decrease monotonically in this ordered list and the only value they share is the IP address of the attack source. This implies that only the latency information can be used to determine the relative position of the gateways on the path between the server and the attack source.

% idea: use traceroute-like protocol

% any gateway on the path can respond
% - some can lie about their control over all traffic from 
% - distinct problem from that of censorship discovery, since here we are interested in finding gateway willing to participate in our protocol

\begin{algorithm*}[h!]
\caption{Gateway discovery algorithm for the server.}
\label{alg:discovery}
\begin{algorithmic}[1]
	\Procedure{GatewayDiscovery}{$\mathcal{A}$, $\mathit{TTL}_\mathit{max}$} \Comment{$\mathcal{A}$ is the set of attack packets.}
	\State $\texttt{GList} \gets \emptyset$.
	\State $t \gets 0$. \Comment{TTL value for the current \texttt{DISCOVER} packet.}
	\State $s \gets \text{one of} \left\{ \text{\it source}(A) \mid A \in \mathcal{A} \right\}$. \Comment{Assume this is a singleton set.}
	\While{$t < \mathit{TTL}_\mathit{max}$} \Comment{$\mathit{TTL}_\mathit{max}$ is the maximum hops to explore.}
		\State Increment $t$.
		\State Send ICMP \texttt{DISCOVER} packet, with TTL $t$, to IP address $s$.
		\State Wait for ICMP response packet. \textbf{break} if timed out.
		\If{ICMP \texttt{GATEWAY-RESPONSE} packet received from IP address $G$}
			\State $\langle \text{latency\ } l, \text{IP range\ } \Pi\rangle \gets $ packet contents.
			\If{$s \in \Pi$}
				\State Add $\langle G, l \rangle$ to $\texttt{GList}$.
			\EndIf
		\EndIf
		\If{ICMP \texttt{ECHO-REPLY} packet received}
			\State \textbf{break}
		\EndIf
	\EndWhile
	\State Sort \texttt{GList} by latencies.
	\State \Return $\left\{ G \mid \langle G, l \rangle \in \texttt{GList} \right\}$.
	\EndProcedure
\end{algorithmic}
\end{algorithm*}

\paragraph{Asymmetric or Dynamic Internet Paths}
We leave for future work the question of discovering the gateways on paths that are asymmetric (i.e., source-to-server uses a different path than server-to-source), or where the paths change rapidly as can be the case for some mobile nodes.

\vspace*{\baselineskip}
We note that the above Algorithm~\ref{alg:discovery} together with the rest of the system do not assume that each gateway is honest in any of its interactions with the server. At this point, we need to consider that some gateways may claim to control larger IP ranges than the ones they really manage, in order to be more likely to participate in the protocol and earn money. This challenge is addressed by authenticating the claims of the chosen gateway and ensuring that the attack does truly originate from a network managed by this gateway.

\subsection{Authentication}
Having discussed the mechanism for gateway discovery, we now discuss details of how the server and the gateway determine if the attack is really coming from the local area behind the gateway or not. Note that even though the attacker may have incentives to spoof the gateway's address in the attack packets, it is important to observe that it is of no use to him if he spoofs the server's address instead. This is because if the server's address is spoofed, the gateway's replies to the spoof check algorithm will go to the spoofed address instead and the protocol will not be able to proceed. Hence, we only talk about spoof detection with respect to the source address in the attack packets.

\begin{algorithm*}[h!]
\caption{Spoof checking algorithm by the server. Checks if the attack is actually coming from behind the gateway. Returns \texttt{true} if $G$ was spoofed.}
\label{alg:serverSpoof}
\begin{algorithmic}[1]
	\Procedure{Server-SpoofCheck}{$G, \mathcal{A}, \Delta$} \Comment{$G$ is the gateway's IP address}
	\State Ping $G$ with an attack packet $A \in \mathcal{A}$ and $\Delta$. \Comment{This is the \texttt{PING} packet.}
	\State Wait for the response from $G$. \Return \texttt{true} if timed out.
	\If{\texttt{VERIFIED} packet received from $G$}
		\State \Return \texttt{false}
	\EndIf
	\If{\texttt{CHECK} packet received from $G$}
		\State Wait for $5\Delta$ time steps for an attack packet $A'$ to be received. \Return \texttt{false} if timed-out.
		\If{$A'$ contains $\texttt{stamp}$}
			\State Send $A'$ to $G$.
			\State Wait for the response from $G$. \Return \texttt{true} if timed out.
			\If{\texttt{VERIFIED} packet received from $G$}
				\State \Return \texttt{false}
			\EndIf
		\EndIf
	\EndIf
	\State \Return \texttt{true}
	\EndProcedure
\end{algorithmic}
\end{algorithm*}

Our algorithms for spoof check are described in Algorithms~\ref{alg:serverSpoof} and~\ref{alg:gatewaySpoof}. The server prepares a special packet, called \texttt{PING}, which contains an attack packet $A$ and the average inter-packet time $\Delta$ as observed by the server. When the gateway receives this packet, it is either able to immediately recognize that it is from its local network, in which case the gateway sends a special packet, called $\texttt{VERIFIED}$ to the server. However, if this is not the case, then the gateway generates a \texttt{stamp}, which is essentially a privately known string to the gateway, and sends a $\texttt{CHECK}$ packet to the server what contains $\texttt{stamp}$. If the server receives a $\texttt{VERIFIED}$ packet, it knows that the gateway's address was not spoofed and its spoof check is complete. Else, the server observes all attack packets for $5\Delta$ time steps and looks for the string $stamp$ in them. The idea is that the gateway is now attaching \texttt{stamp} to each packet that it forwards to the server for $5\Delta$ time steps. Thus, if the attack really came from behind the gateway, the server must receive an attack packet with this \texttt{stamp} attached. 

If the server intercepts a packet with the \texttt{stamp}, it immediately sends the packet to the gateway and the gateway replies with a \texttt{VERIFIED} packet. The learning phase can now start. However, if no such packet is received, the server concludes that the attack did not come from the gateway's local network and removes that gateway from $\texttt{GList}$. Our assumption here is that one of the factors that distinguishes an attack packet from a good packet is the frequency of packet arrival. If the attacker did not send any attack packets for $5\Delta$ time steps, then essentially the frequency of attack packets is reduced five fold. This number, hence, can be set to whatever value that seems sufficient for the server to stop classifying the packets as a attack packets. 

\begin{algorithm*}[h!]
\caption{Spoof checking algorithm by the gateway. Checks if the packet is actually coming from its local network.}
\label{alg:gatewaySpoof}
\begin{algorithmic}[1]
	\Procedure{Gateway-SpoofCheck}{$S, A, \Delta$}
		\If{source of $A$ identified to be from local network} \Comment{Through trivial checks}
			\State Send $\texttt{VERIFIED}$ to the server.
			\State \Return \texttt{false}
		\Else
		 	\State Generate $\texttt{stamp}$.
		 	\State Send $(\texttt{CHECK}, \texttt{stamp})$ to the server.
		 	\State Attach \texttt{stamp} to all outgoing packets to the server for $5\Delta$ time steps.
		 	\State If no response from the server in this time, \Return \texttt{true}.
		 	\If{invalid $\texttt{stamp}$ in the packet}
		 		\State \Return \texttt{true}
		 	\Else
		 		\State Send $\texttt{VERIFIED}$ to the server.
		 		\State \Return \texttt{false}
		 	\EndIf
		\EndIf
	\EndProcedure
\end{algorithmic}
\end{algorithm*}

Another important assumption we make here is that the \texttt{stamp} is irreproducible by the attacker. This can be done by the establishment of an authenticated channel between the server and the gateway prior to spoof check, or through any mechanism that prevents the attacker from obtaining this string. For now, we leave the details of establishment of this \texttt{stamp} efficiently as future work and refer readers to existing solutions for sending private information over the communication link. 

\subsection{Initial Reward Selection}
We now discuss how the server and the gateway decide on the initial reward for the \remoteGate protocol. The selection of this reward is critical to the incentive-compatible nature of our protocol since this is the reward that essentially incentivizes the gateway to participate in the protocol with the server at the cost of using its computation resources and modifying its firewall. Since the server and the gateway are both modeled as rational players, we design a negotiation scheme which is dominant-strategy incentive compatible. In other words, we design a scheme in which the dominant strategy for both the server and the gateway is to bid their true valuations for what the protocol. Algorithms~\ref{fig:serverInit} and~\ref{fig:gatewayInit} provide detailed steps for this negotiation. 

The problem of deciding the initial reward falls under the broader category of bilateral trade, introduced in the seminal work of Myerson and Satterswaite~\cite{myerson1983efficient}. In this setting, a single seller is the owner of an indivisible item which can be traded with a single buyer. In our setting, the gateway is the seller, the model it learns is the \emph{item} being sold and the server is the buyer. Apart from (dominant-strategy) incentive compatibility (DSIC), bilateral trade problems typically also try to ensure two more properties: (1) \emph{individual rationality}, which in our case translates to the fact that the server and gateway participate voluntarily in the protocol and can leave anytime if incentivized to do so, and (2) \emph{budget balance}, which in our case means that neither the server and the gateway suffer losses in terms of the rewards paid or received. A recent result by~\cite{colini2016approximately} shows that no dominant-strategy mechanism which is also individually rational and budget balanced can guarantee more than $0.749$ of the optimal welfare. Since we don't assume any shared randomness between the server and the gateway or any shared knowledge of the distribution of initial rewards, we use a deterministic take-it-or-leave-it strategy, also called \emph{fixed-price bilateral trade}, as devised in~\cite{blumrosen2016almost} to achieve $0.5$ of the optimal welfare, which is optimal for any deterministic strategy (Proposition $3.3$ in~\cite{blumrosen2016almost}).

Let $v_S$ be the server's true valuation/belief of how much the gateway should be paid per correctly classified example in the first round, based on factors like how important the prevention of the attack is as well as an estimate of the computational efforts used by the gateway to generate a sufficiently accurate model. Similarly, let $v_G$ be the gateway's true valuation of the reward it should receive. The authors would like to refer the reader to Appendix~\ref{app:paramter} for details on how $v_S$ and $v_G$ can be chosen in practice. Let $\rho_1$ be the reward chosen in this initial negotiation phase. This reward, once agreed upon, will be paid for the accuracy reported by the gateway on the test examples that the server sends in the first round of the learning phase. Then, we define the utility function of the server to be $v_S - \rho_1$, since it is trying to negotiate as small a reward as it can. Similarly, we define the gateway's utility as $\rho_1$, since it is trying to maximize the reward it collects. Based on the setting just defined, we choose the reward by setting  $\rho_1 = \min \{v_S, v_G\}$, which becomes $\rho_1 = v_G$ assuming $v_G \leq v_S$.\footnote{If $v_S < v_G$, then the negotiation cannot happen assuming that the server will not pay more than $v_S$ and the gateway will not accept a payment less than $v_G$.} From Theorem $3.1$ in~\cite{blumrosen2016almost}, this mechanism guarantees DSIC, individually rationality, budget balance and achieves at least $1/2$ of the optimal social welfare.\footnote{The result in~\cite{blumrosen2016almost} holds for any distributions from which $v_S$ and $v_G$ are chosen. In particular, it is true for point distributions.}

\begin{algorithm*}[h!]
\caption{Initialization for the server. Called when \textsc{Server-SpoofCheck}$(\mathcal{A})$ succeeds.}
\label{fig:serverInit}
\begin{algorithmic}[1]
	\Procedure{Server-StartService}{$\varepsilon, \mathcal{A}, v_S, \Gamma_\text{service}^S$}
		\State Sample $A \in \mathcal{A}$ at random.
		\State Send $(\texttt{COMMIT}(v_S), \Gamma_\text{service}^S, A, \texttt{COMMIT}(\varepsilon))$ to the gateway. \Comment{This is the \texttt{INIT} packet.} 
		\State Wait for the response from the gateway. \Return $\texttt{false}$ if timed out.
		\If{$\Gamma_\text{service}^G$ or $\texttt{fee}$ or $\iota$ not acceptable}
			\State \Return $\texttt{false}$
		\Else
			\State $\rho_1 \gets v_G$
			\If{$\rho_1 > v_S$}
				\State  \Return $\texttt{false}$
			\Else
				\State Send $\paran{\rho_1, \texttt{OPEN}(v_S)}$.
				\State \Return $\Call{Server-Learning}{\mathcal{A}, \rho_1, \varepsilon, r_{\text{max}},\Gamma_\text{service}^G, \texttt{fee}, \iota}$
			\EndIf
		\EndIf
	\EndProcedure
\end{algorithmic}
\end{algorithm*} 

Having understood the initial reward selection mechanism, we now discuss the steps in Algorithms~\ref{fig:serverInit} and~\ref{fig:gatewayInit} in detail. The server begins with sampling an attack packet $A$ at random the set $\mathcal{A}$ of all attack packets. It uses a commitment scheme to commit the value of $v_S$ to the gateway and sends along the packet $A$, value of the service duration period $\Gamma_\text{service}^S$ and a commitment of the error tolerance $\varepsilon$. The two commitments are made for the gateway to be able to verify the server's computation of $\rho_1$ and the total reward later in the protocol. Upon receiving these values from the server, the gateway decides on the length $\Gamma_\text{service}^G$ of the period it can actually provide service for, the \texttt{fee} it wishes to charge the server for the service once the model is deployed in the firewall, the number $\iota$ of post-paid installments it is willing to accept for the payment during the service period, and the value of $v_G$ based on the estimates of computation and resources that will be used to learn the model. Upon decision of these parameters, the gateway sends them to the server and wait for a response.

\begin{algorithm*}[h!]
\caption{Initialization for the gateway.}
\label{fig:gatewayInit}
\begin{algorithmic}[1]
	\Procedure{Gateway-StartService}{$\texttt{COMMIT}(\rho_{S,1}), \Gamma_\text{service}^S, A$}
		\State Initialize $\Gamma_\text{service}^G, \texttt{fee},\iota$ and $v_G$ based on $A$ and $\Gamma_\text{service}^S$. \Comment{$\texttt{fee}$ is (post)paid in $\iota$ installments.}
		\State Send $(\Gamma_\text{service}^G, v_G, \texttt{fee},\iota)$. 
		\State Wait for the response from the server. \Return if timed out.
		\If{$\rho_1 < v_G$ (or computed incorrectly) or $\texttt{OPEN}(v_S)$ does not verify}
			\Return
		\EndIf
		\State $\Call{Gateway-Learning}{\rho_1, \Gamma_\text{service}^G, \texttt{fee}, \iota}$
		\EndProcedure
\end{algorithmic}
\end{algorithm*}

If the server determines that the service period proposed by the gateway is non-acceptable (may be too short or too long), or that the \texttt{fee} is too high or that the number of installments is unreasonable (typically too low), then it can discontinue the protocol with the gateway and move to the next gateway, if available, in $\texttt{FList}$. Else, the server computes the reward $\rho_1$ as the value $v_G$ that the gateway proposed (for the reasons mentioned above) and verifies that the payment is not more than its budget $v_S$. The server reports the gateway of the reward $\rho_1$ selected and opens its commitment for $v_S$ for the gateway to be able to verify the computation of $\rho_1$. It then proceeds to beginning the learning phase with the gateway. Upon receiving the server's message, if the gateway is able to verify the computation of $\rho_1$ and that it gets paid at least $v_G$, it proceeds to beginning the learning phase with the server. 

\subsection{Rule Learning}
We now discuss the learning phase of the \textsc{RemoteGate} protocol that allows the server to stop the flow of attack packets it is receiving by helping the gateway come up with a classification rule that can be deployed in the gateway's firewall. This rule (or model) must filter out the attack packets (up to the server's error tolerance) and prevent the attacker from sending more such packets to the server. 

We present the details of the learning phase in Algorithms~\ref{fig:serverAlgo} and~\ref{fig:gatewayAlgo}. These algorithms proceed in rounds, indexed by $r \geq 1$. In each round, the server sends the gateway some training examples, $\texttt{TRAIN}_r$, for the gateway to train its model, $\texttt{MODEL}_r$, and also some test examples, $\texttt{TEST}_r$, for the gateway to apply its model and reply back to the server which the classifications, denoted $\texttt{LABEL}_G\paran{\texttt{TEST}_r}$, it obtains.\footnote{The server poses a hard limit (unknown to the gateway) on the number of rounds $r_\text{max}$ it is willing to engage in with the gateway before giving up and switching over to a different gateway (the next one in $\texttt{FList}$, if any). This limit can be decided based on the number of attack packets available with the server as well as the severity of the attack and the urgency of its prevention as seen fit by the server. It can also be dependent on the total reward the server is willing to offer the gateway as well as the maximum time the server is willing to engage in this protocol for.} The gateway uses these training examples to train a model that correctly classifies the attack packets (best-effort) and then provides the server with what the model predicts on the examples in the test set. Since it can take some time for the gateway to train its model and it is not in the best interest for the server to wait indefinitely for the response, we can assume that the server and the gateway can agree on the maximum time the gateway will engage in training the model, similar to the negotiation of the initial reward. With respect to this, we assume that the gateway, if playing honestly, performs best effort training using whatever learning algorithm it seems fit for the classification. We leave the details of deciding a good time limit for training and the choice of the learning algorithm to future work. 

The server verifies the response from the gateway against her knowledge of the true labels on the test set. It begins by computing the accuracy, $\texttt{acc}_S\paran{\texttt{TEST}_r}$, of the labels that the gateway has provided. If this accuracy is less than the server's error tolerance, $\varepsilon$, it pays the gateway the reward, $\texttt{REWARD}_r$, and terminates the protocol. The gateway, upon the receipt of the reward, deploys the model in its firewall. However, if the server is not satisfied with the gateway's accuracy, it sends more training examples (and a test set) in the next round for the gateway to retrain the model on more labeled examples, with the hope that this retraining will improve the classification accuracy. The server keeps accumulating the rewards for the rounds until it is satisfied with the gateway's accuracy (or the round limit is reached, in which case it terminates the protocol without paying the gateway any reward) and then pays the gateway the total reward accumulated so far. Note that even after the payment has gone through, the server must ensure that the gateway deployed the correct model as promised. We discuss the details of this deployment verification in the next subsection.

The scheme described above works well when both the server and the gateway are honest. However, since we assume that both the server and the gateway are rational instead, we update the rewards in each round in a clever way that forces both the server and the gateway to act honest. For the gateway, we do this by ensuring that any deviation from the honest strategy will only decrease the reward the gateway can expect to receive. For the server, we achieve our goal by providing the gateway with the actual labels of the test packets in the round following the one they were sent in, so the server is dis-incentivized to lie on the accuracy and hence, the reward. The result is a carefully crafted protocol that both the server and the gateway will run in order to incentivize honest model generation to filter the attack packets.\footnote{Recall that in this paper, we do not allow for the possibility of the gateway to collude with the attacker and rationalize at the same time and rather assume that it indeed wants to help the server in preventing the attack but is being greedy on the reward at the same time.}

\subsubsection{Server's algorithm}

\begin{algorithm*}[h!]
\caption{Server's algorithm.}
\label{fig:serverAlgo}
\begin{algorithmic}[1]
	\Procedure{Server-Learning}{$\mathcal{A}, \rho_1, \varepsilon, r_\text{max}, \Gamma_\text{service}, \texttt{fee}, \iota$}
		\State $r \gets 1$
		\State $\texttt{acc}_S(\texttt{TEST}_{0}) \gets 0$
		\State $\texttt{LABEL}_S(\texttt{TEST}_{0}) \gets \emptyset$
		\State $\texttt{REWARD}_{0} \gets 0$ 
		\State $\rho_{S,1} \gets \rho_1$
		\While{$r \leq r_{max}$}
			\State Initialize $\paran{\texttt{TRAIN}_r, \texttt{LABEL}_S\paran{\texttt{TRAIN}_r}}$ and $\paran{\texttt{TEST}_r,  \texttt{LABEL}_S\paran{\texttt{TEST}_r}}$ from $\mathcal{A}$.
			\If{$r \geq 2$}
				\State $\rho_{S,r} \gets \paran{\frac{\sum_{i=1}^{r-1} \texttt{acc}_S\paran{\texttt{TEST}_i}}{2\sum_{i=1}^{r} \texttt{acc}_S\paran{\texttt{TEST}_i}}} \rho_{S,r-1}$
			\EndIf
				\State Send $\paran{\texttt{COMMIT}\paran{\rho_{S,r}}, \texttt{TRAIN}_r, \texttt{LABEL}_S(\texttt{TEST}_{r-1})}$ to the gateway.
				\State Wait for the response from the gateway. \Return \texttt{false} if timed out.
				\State $\rho_r \gets \min \paran{\rho_{S,r}, \rho_{G,r}}$
				\State Send $\paran{\rho_r, \texttt{OPEN}\paran{\rho_{S,r}}, \texttt{TEST}_r}$ to the gateway. 
				\State Wait for the response from the gateway. \Return \texttt{false} if timed out.
			\State $\texttt{acc}_S(\texttt{TEST}_{r}) \gets \ \mid \{ i \mid \texttt{LABEL}_S(\texttt{TEST}_{r}) = \texttt{LABEL}_G(\texttt{TEST}_{r}) \} \mid$
			\State $\texttt{REWARD}_r \gets \rho_r \sum_{i=1}^r \texttt{acc}_S(\texttt{TEST}_{i})$
			\If{$\texttt{acc}_S(\texttt{TEST}_{r}) \geq (1-\varepsilon)|\texttt{TEST}_{r}|$}
				\State Send $\texttt{OPEN}(\varepsilon)$ and initiate payment of $\texttt{REWARD}_r$ to the gateway.
				\If{payment was successful}\label{server:payment}
					\State \Return \Call{Server-Deployment}{$G, \Gamma_\text{service}, \texttt{fee}, \iota$}
				\EndIf
			\Else
				\State $r \gets r+1$
			\EndIf
		\EndWhile
		\State \Return \texttt{false}
	\EndProcedure
\end{algorithmic}
\end{algorithm*}

The server's algorithm is presented in detail in Algorithm~\ref{fig:serverAlgo}. It begins with initializing some boundary conditions: \textit{a priori} accuracy $\texttt{acc}_S\paran{\texttt{TEST}_0}$ on the test set $\texttt{TEST}_0$ to zero; the set of labels $\texttt{LABEL}_S\paran{\texttt{TEST}_0}$ on this test set to the empty set; and the cumulative reward so far, $\texttt{REWARD}_0$, to zero as well. The rounds then begin to progress and at the beginning of each round, the server makes sure that the maximum round index has not been exceeded.

In each round, the server begins by arranging the attack packets along with some good packets into two sets : one that it wishes to send to the gateway for training the model, and the other for testing. We typically assume that the training as well as the testing set is obtained by independent sampling without replacement from the set of all attack and good packets. This can even be done in a way that the training (testing) sets "appear" to be chosen uniformly at random to the gateway by local coin flipping at the server's end (put an attack packet if heads and good packet if tails). Maintaining the overall quality of the function constructed by the learning algorithm (expected loss, with respect to the underlying distribution)~\cite{dekel2010incentive} makes this construction necessary.

After the selection of training and test packets, the server computes the reward $\rho_{S,r}$ it wishes to pay to the gateway per correctly classified example in the test set for this round. The reward is chosen based on a bidding mechanism, where both the server and the gateway provide their bids, $\rho_{S,r}$ for the server and $\rho_{G,r}$ for the gateway, for the reward and the smaller of the two is chosen as the reward $\rho_r$ for this round. This is commensurate with the bilateral trading scheme as discussed during the initial reward selection phase and provides a DSIC strategy for the server and the gateway to act honestly during each round of the protocol. A commitment scheme is used here for the gateway and server to be able to carry out this blind negotiation and also verify the computation of $\rho_r$.

The server also sends the training set $\texttt{TRAIN}_r$ and the labels $\texttt{LABEL}_S\paran{\texttt{TEST}_{r-1}}$ for the test set in the previous round. The latter is to assist the gateway in further training the model on more labeled examples with the hope that this will increase the model's accuracy. Once the classification $\texttt{LABEL}_G\paran{\texttt{TEST}_r}$ on the tests labels is obtained from the gateway, the server computes the accuracy $\texttt{acc}_S\paran{\texttt{TEST}_r}$ by comparing them against the true labels $\texttt{LABEL}_S\paran{\texttt{TEST}_r}$. If the accuracy was more than the server's threshold $\varepsilon$, then the server open its commitment for $\varepsilon$ and initiates the payment protocol to pay the reward to the gateway. The total reward is calculated by multiplying the per-example reward in the last round of the protocol with the total number of correctly classified test examples so far. If this payment is successful, the server assumes the beginning of the service period as negotiated previously. The details of what happens during this service period will be explained in later subsections.

If, however, the accuracy in the current round is unsatisfactory, the server proceeds to the next round, where it starts by adjusting the reward. The server's bid for the reward changes based on the accuracy reported by the gateway in the preceding rounds. However, to prevent the gateway from increasing the number of rounds in order to collect more reward, the server reduces its bid by half for the reward in every round according to the following equation: \[ \rho_{S,r} = \paran{\frac{\sum_{i=1}^{r-1} \texttt{acc}_S\paran{\texttt{TEST}_i}}{2\sum_{i=1}^{r} \texttt{acc}_S\paran{\texttt{TEST}_i}}} \rho_{S,r-1} \] Here, the $2$ in the denominator can be replaced by any constant $c > 1$. Rearranging the terms of this equation gives $\texttt{REWARD}_r = \texttt{REWARD}_{r-1}/2$, implying that the total reward that the gateway collects gets reduced to half every time it is unable to learn the model to sufficient accuracy. To prevent the server from underpaying an honest gateway, we assume that $v_S$ was chosen keeping in mind this reduction in the total reward.

 The gateway uses a similar scheme to update its bid for the reward and reports it to the server. The two bids must match in every round for both the server and the gateway.  

\subsubsection{Gateway's algorithm}

\begin{algorithm*}[h!]
\caption{Gateway's algorithm.}
\label{fig:gatewayAlgo}
\begin{algorithmic}[1]
	\Procedure{Gateway-Learning}{$\rho_1, \Gamma_\text{service}, \texttt{fee}, \iota$}
		\State $r \gets 1$
		\State $\texttt{acc}_G(\texttt{TEST}_{0}) = 0$
		\State $\texttt{MODEL}_{0} \gets \emptyset$
		\State $\rho_{G,1} \gets \rho_1$
		\While{\texttt{true}}
			\If{$r \geq 2$}
				\State $\texttt{acc}_G(\texttt{TEST}_{r-1}) \gets \ \mid \{ i \mid \texttt{LABEL}_S(\texttt{TEST}_{r-1}) = \texttt{LABEL}_G(\texttt{TEST}_{r-1}) \} \mid$
				\State $\rho_{G,r} \gets \paran{\frac{\sum_{i=1}^{r-1} \texttt{acc}_G\paran{\texttt{TEST}_i}]}{2\sum_{i=1}^{r} \texttt{acc}_G\paran{\texttt{TEST}_i}}} \rho_{G,r-1}$
			\EndIf
			\State Send $\rho_{G,r}$ to the server.
			\State Wait for the response from the server. \Return \texttt{false} if timed out.
			\State Verify the commitment for $\rho_{S,r}$ and the computation of $\rho_r$. \Return \texttt{false} if not.
			\State $\texttt{MODEL}_r \gets \texttt{MODEL}_{r-1}$ trained on $\texttt{TRAIN}_r$ and $\paran{\texttt{TEST}_{r-1}, \texttt{LABEL}_{S}\paran{\texttt{TEST}_{r-1}}}$.
			\State $\texttt{LABEL}_G \paran{\texttt{TEST}_r} \gets \texttt{MODEL}_r \paran{\texttt{TEST}_r}$
			\State Send $\texttt{LABEL}_G$ to the server. 
			\State Wait for the response from the server. \Return \texttt{false} if timed out.
			\If{payment initiated by the server}
				\If{$\texttt{OPEN}(\varepsilon)$ fails to verify}
					\State \Return \texttt{false}
				\Else
					\For{$r' = 1,\dots,r$} \Comment{Check if server did not lie on the accuracies.}
						\If {$\texttt{LABEL}_S \paran{\texttt{TEST}_{r'}}$ differs from $\texttt{MODEL}_r \paran{\texttt{TEST}_{r'}}$ by more than $\varepsilon$}
							\State \Return \texttt{false}
						\EndIf
					\EndFor
				\EndIf
				\If{$\texttt{REWARD}_r$ received from the server}\label{server:payment}
					\State Deploy $\texttt{MODEL}_r$ in firewall.
					\State \Return \Call{Gateway-Deployment}{$S, \Gamma_\text{service}, \texttt{fee}, \iota$}. \Comment{$S$ refers to the server's identity.}
				\EndIf
			\Else 
				\State $r \gets r+1$
				\State Wait for the response from the server. \Return \texttt{false} if timed out.
			\EndIf
		\EndWhile
		\State \Return \texttt{false}
	\EndProcedure
\end{algorithmic}
\end{algorithm*}    

We now discuss the details of Algorithm~\ref{fig:gatewayAlgo} to describe the steps that the gateway performs during the learning phase. After initializing the protocol parameters similar to what the server did, the gateway produces its bid $\rho_{G,r}$ to the server. Note that this bid is the same as the initial reward chosen previously. Upon receiving the server's response, the gateway verifies if the commitment for $\rho_{S,r}$ was correctly made and the server correctly computed $\rho_r$. If the test passes, the gateway (re)-trains the model $\texttt{MODEL}_r$ based on the training examples from the current round and the test examples from the previous round (along with their labels received in this round). It then computes the classification $\texttt{LABEL}_G\paran{\texttt{TEST}_r}$ using $\texttt{MODEL}_r$ and sends it to the server. 

If the server replies with the commitment-open for $\varepsilon$, the gateway verifies if $\texttt{REWARD}_r$ is commensurate with the classifications $\texttt{MODEL}_r$ would have provided for all the previous rounds. In other words, since the server is now trying to minimize the reward it gives to the gateway, lying on the accuracy of the test examples seems advantageous to the server since that will allow the server to lie on the reward accumulated. This way, once the final round happens where the server agrees on the model to be sufficiently accurate, the gateway can run all the previous test examples against the model and check if the labels provided by the server were consistent with what this model now reports (within $\varepsilon$ parameter, which the server is now required to commit at the beginning of the protocol and open in the last round, just before payment happens). If the gateway finds any inconsistencies, it immediately recognizes treachery on the part of the server and aborts the protocol. Furthermore, if the server lies consistently, then the gateway may not be able to learn the model at all and the server would have wasted its time and not get the attack packets stopped. Even if the gateway learnt the model somehow, it is likely to have overfitted the model, which makes it likely that some false positives are blocked after deployment or some false negatives cause attack packets to reach the server. This is detrimental to the server since the attack will not be prevented in this case (as no model will be deployed in the firewall). Thus, it is not in the server's best interest to lie on the accuracies reported, and hence, the reward accumulated, to the gateway.

Once the gateway receives the reward, it deploys the model $\texttt{MODEL}_r$ in its firewall and proceeds to the service period. However, if the server replies with more training and test examples, the gateway updates its bid for the reward based on the labels provided by the server. It then retrains it model and provides classifications on the new test examples. The rounds progress until the server is satisfied with the gateway's model.

An interesting way for the gateway to cheat during the bidding phase for $\rho_r$ is to carefully manipulate the classifications $\texttt{LABEL}_G$ on the test sets so that the total reward collected is more than what it would have been if the gateway was playing honestly. The gateway can try to achieve this by temporarily reporting the classifications so that the accuracy calculated is low. The motivation behind this is to force the server into more and more rounds, thus collecting a high total reward than a smaller reward in the fewer rounds. However, it still has to do it in way that the server does not exhaust its maximum round limit $r_\text{max}$. 

We show that it is not possible for the gateway to obtain a reward higher than what it would have collected upon honest play. In fact, a stronger statement can be said here about the gateway's strategy. Not only does following the honest strategy maximize the overall reward the gateway can collect from the server, but this holds true in each round as well. The following lemma establishes the proof.

\begin{lemma}
	Let $r^*$ be the number of rounds that the gateway takes to train the model if it plays honestly. Then, the maximum reward the gateway can collect from the server is at most $\texttt{REWARD}_{r^*}$.
\end{lemma}
\begin{proof}
	Let $r$ be the actual number of rounds taken by the gateway to perform the training. Clearly, $r \geq r^*$, since otherwise, the honest learning would finish in $r$ rounds and not $r^*$. Let $r = r^* + \zeta$ for some $\zeta \geq 0$. Then, we have the following.
	\begin{equation*}
		\begin{split}
			\texttt{REWARD}_r &\leq \rho_{S,r} \sum_{i=1}^r \texttt{acc}_S(\texttt{TEST}_{i})\\
			&= \paran{\frac{\sum_{i=1}^{r-1} \texttt{acc}_S\paran{\texttt{TEST}_i}}{2\sum_{i=1}^{r} \texttt{acc}_S\paran{\texttt{TEST}_i}}} \rho_{S,r-1} \sum_{i=1}^r \texttt{acc}_S(\texttt{TEST}_{i}) \\
			&= \frac{1}{2} \rho_{S,r-1} \sum_{i=1}^{r-1} \texttt{acc}_S\paran{\texttt{TEST}_i} \\
			&= \frac{1}{2^\zeta} \rho_{S,r^*} \sum_{i=1}^{r^*} \texttt{acc}_S\paran{\texttt{TEST}_i} = \frac{1}{2^\zeta} \texttt{REWARD}_{r^*} \\
			&\leq \texttt{REWARD}_{r^*}
		\end{split}
	\end{equation*} 
	Hence, the reward up to round $r$ can never exceed the reward collected by optimal play.
\end{proof}

This lemma proves that the dominant strategy for the gateway is to play honestly in each round during the learning phase. Combining this with our discussion of the initial reward selection and how the logs maintained by the gateway and the commitment schemes force the server to act honest, we have shown that the \remoteGate protocol is incentive compatible. Note that although we do not explicitly mention this, but the server's utility function contains an additive term which represents the need for the attack to be prevented. This dis-incentivizes the server to take any action that will make the gateway abort the protocol and not deploy the model in its firewall, since then the attack will not be prevented. We assume that the amount of total reward and fee paid by the server is small compared to the what the server gains when the attack is successfully prevented.

\subsection{Payment Protocol}
Once the learning phase is over, the server pays the gateway its promised reward. Also, during the service phase, the server periodically pays the gateway a \texttt{fee} in appraisal of the retention of the model in the firewall. In this section, we abstract the underlying payment scheme used for all these payments. The exact details are out of the scope for this paper, although, specialized payment schemes for this protocol pose interesting open problems for future work.

\begin{figure*}
\begin{center}
	\includegraphics[height = 2.5cm, width = 11cm]{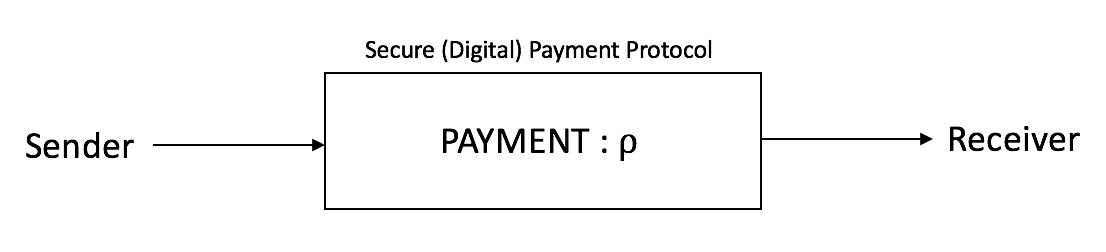}
	\end{center}
	\label{fig:payment}
	\caption{Secure digital payment protocol blackbox. This is used to send $\rho$ units of money from the Sender to the Receiver.}
\end{figure*}

Our algorithm is flexible in the use of any trusted digital payment protocol, as long as it is secure and provable. By \emph{trusted}, we mean that when the server is paying the gateway its reward, the payment protocol must ensure that the money goes from the server's account to an account that belongs to the gateway and no one else. Moreover, we want this protocol to be \emph{secure} so that any data that is passed between the server/gateway to the merchant helping with the payment must be protected against tampering by malicious third parties. We also require that the payment protocol be \emph{provable}, which means that at any point after the protocol terminates, both the server and the gateway must be able to provide a proof of the success or failure of the payment along with the amount of the reward in question. If an invalid or false proof is provided, then the payment protocol must be able to identify this and report accordingly. Some examples of payment protocols that can be used here are Visa Checkout~\cite{visa}, PayPal~\cite{paypal}, Google Wallet~\cite{google}, Apple Pay~\cite{apple} etc.

\subsection{Deployment Verification}
Once the gateway has convinced the server of an accuracy that is at least $\varepsilon$, the server initiates the payment protocol (line~\ref{server:payment} of Algorithm~\ref{fig:serverAlgo}). The important thing here is for the gateway to now deliver its promise of deploying the model by actually installing it in the firewall and making sure no future attack packets reach the server. However, since the gateway is rational, it needs to be incentivized to do this. As mentioned previously, the incentive is in the form of a periodic \texttt{fee} that the server pays to the gateway to maintain the model in its firewall. The details of the deployment verification phase are provided in Algorithms~\ref{alg:serverDeployment} and~\ref{alg:gatewayDeployment}. 

\begin{algorithm*}[h!]
\caption{Server's protocol for deployment.}
\label{alg:serverDeployment}
\begin{algorithmic}[1]
	\Procedure{Server-Deployment}{$G, \Gamma_\text{service}, \texttt{fee}, \iota$}
	\For{$i = 1,\dots,\iota$} 	
		\State Observe traffic for $(\Gamma_\text{service}/\iota)$ time steps.
		\If{attack packets $A_1,A_2,\dots$ received with average interval $\Delta' < \Delta / \varepsilon$}
			\State $\mathcal{A} \gets \{ A_1,A_2,\dots \}$
			\State $\texttt{GList} \gets \Call{GatewayDiscovery}{\mathcal{A}}$
			\If{$G \in \texttt{GList}$}
				\If{$\Call{Server-SpoofCheck}{G,\mathcal{A},\Delta'} = \texttt{false}$}
					\State $\texttt{LOG}_S \gets $ "\texttt{BREACH}"
					\State \Return \texttt{false}.
				\EndIf
					\State Monitor $\mathcal{A}$ for taking preventive measures in the future, if required.
				\EndIf
			\EndIf
		\State Pay \texttt{fee} to gateway $G$.
	\EndFor
	\State \Return \texttt{true}
	\EndProcedure
\end{algorithmic}
\end{algorithm*}

In Algorithm~\ref{alg:serverDeployment}, the server begins the deployment verification phase, also referred to as \emph{service}, by observing the traffic for certain time steps and accordingly paying the gateway for that time period. More concretely, since the total service duration agreed upon is $\Gamma_\text{service}$ and the number of installments in which the payment will be made is $\iota$, the server observes the traffic in time slots that are $\Gamma_\text{service} / \iota$ time steps long. In each slot, the server checks if an attack packet was received or not. Since the server allowed an error tolerance of $\varepsilon$ for the gateway's model, it now expects the attack frequency smaller by a factor of $\varepsilon$. Thus, if the server receives attack packets that are at least $\Delta/\varepsilon$ time steps apart\footnote{recall that $\Delta$ was the original inter-packet time in $\mathcal{A}$}, then it is convinced that the model is correctly installed in the gateway's firewall and hence, it pays the gateway the agreed upon \texttt{fee}. 

(\textbf{Note 1}: The discussion above provides a way for the server to decide what value of $\varepsilon$ to set. Ideally, the server would like $\Delta/\varepsilon > \Gamma_\text{service}$. Since the only value known to the server at the beginning of the algorithm is $\Gamma_\text{service}^S$, the server can set $\varepsilon$ to any value less than $\Delta / \Gamma_\text{service}^S$ as a good estimate. Another way would be to modify Algorithm~\ref{fig:serverInit} slightly so that the commitment for $\varepsilon$ goes after $\Gamma_\text{service}$ is decided. This way, the server can set $\varepsilon$ to any value less than $\Delta / \Gamma_\text{service}$ instead and be sure that during the service period, no\footnote{at most one to be precise.} attack packets must be received.) 

If, however, the server receives attack packets that arrive more frequently, the server knows that either the attacker is spoofing the gateway's address in the new attack packets or the gateway never deployed the (or deployed a different) model in its firewall. To find what case it is, the server takes similar steps as it took in Algorithm~\ref{alg:remoteGateServer}. It begins with running the gateway discovery protocol in Algorithm~\ref{alg:discovery} to determine the list of gateways that are closer to the attack source. If the gateway $G$ with whom the service is going on is not part of this list, then the server knows that the attacker has changed its source of attack, but $G$ is fulfilling its promised protection. Hence, it pays $G$ the \texttt{fee} as promised.\footnote{In this case, ideally, the server should should take more pressing measures to stop the attack.} If $G$ was indeed part of the list, the server runs the spoof-check protocol in Algorithm~\ref{alg:serverSpoof} with $G$ to ensure that its address was not spoofed. If it was spoofed, the server pays the gateway its promised \texttt{fee}, else, it identifies that the gateway is acting malicious by not fulfilling its promise and hence, immediately terminates any future payments to the gateway. In addition, the server initiates conflict resolution with the trusted third party to claim its reward and the already paid fee back.

(\textbf{Note 2}: Similar to the discussion in note $1$, Algorithm~\ref{alg:serverDeployment} also suggests a way for the server to decide what $\iota$ to agree upon. Ideally, the $\iota$ suggested by $G$ in Algorithm~\ref{fig:serverInit} is such that $\Gamma_\text{service}/\iota$ is more than the time it took for the server to run the spoof-check with $G$, say $\tau$, so that during this deployment verification phase, the server can safely run the spoof check within one installment slot. Thus, the server accepts $\iota$ only when it is less than $\Gamma_\text{service}/\tau$. Combine this with note $1$ above to get $\iota = O(1/\varepsilon)$. Moreover, since the gateway also has a good estimate of what $\tau$ is, it is in its best interest to suggest $\iota$ accordingly or else it risks the server terminating the protocol and switching to a different gateway.) 

\begin{algorithm*}[h!]
\caption{Gateway's protocol for deployment.}
\label{alg:gatewayDeployment}
\begin{algorithmic}[1]
	\Procedure{Gateway-Deployment}{$S, \Gamma_\text{service}, \texttt{fee}, \iota$}
	\For{$i = 1,\dots,\iota$} 	
		\If{$(\Gamma_\text{service}/\iota)$ time steps have passed}
			\If{\texttt{fee} not received}
				\State $\texttt{LOG}_G \gets $ "\texttt{BREACH}"
				\State \Return \texttt{false}.
			\EndIf
			\EndIf
	\EndFor
	\State \Return \texttt{true}
	\EndProcedure
\end{algorithmic}
\end{algorithm*}

Having discussed the server's view of verifying deployment, the gateway's protocol in this case is described in Algorithm~\ref{alg:gatewayDeployment}. In each time slot of $\Gamma_\text{service} / \iota$ that passes, the gateway either receives a payment of $\texttt{fee}$ from the server (at the end of the slot) or not. Everything is good if it does, but if the server does not pay the fee for any reason (attack received or other reason), the gateway construes this as breach of contract from the server and resorts to the trusted third party for conflict resolution to demand its money from the server.

We claim that this mechanism provides both the server and the gateway the required incentive to participate, engage and honestly act for the entire duration of the \remoteGate protocol. Specifically, the gateway now has no incentive to not deploy the model, or deploy the wrong model, or discontinue the deployment after some time (which may even be immediately after the payment). 

\subsection{Conflict Resolution}

\begin{figure*}[h!]
\begin{center}
	\includegraphics[height = 2.5cm, width = 10cm]{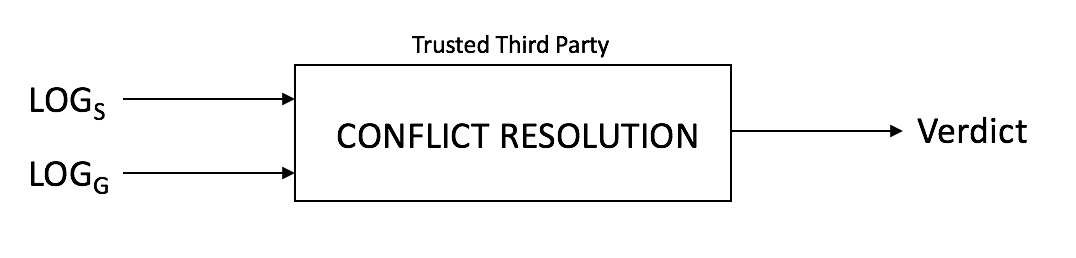}
	\end{center}
	\label{fig:conflict}
	\caption{Schematic of the conflict resolution blackbox, where a trusted third party uses $\texttt{LOG}_S$ and $\texttt{LOG}_G$ to determine a verdict on who acted malicious and take an appropriate action.}
\end{figure*}

The final piece of our protocol is the conflict resolution phase in which a trusted third party is approached by either the server or the gateway to resolve the breach of contract as viewed from either or both parties' perspective(s). The logs of the communication maintained by both the server and the gateway will help this third party to determine who caused the breach and take appropriate actions. The schematic in Fig.~\ref{fig:conflict} depicts this resolution process. For the purpose of our algorithm, the exact details of this conflict resolution are not important and hence, we omit this discussion here. However, we do emphasize that the existence of a trusted third party is only one of the many ways to resolve conflicts. We leave it as an interesting open problem to design an efficient resolution mechanism that does not involve any third party and is scalable as well as trusted.

%\section{Evaluation}
%\label{sec:evaluation}
%\todo[inline]{Mention details of the experiment in mind}
%\todo[inline]{Resource Costs}
%\todo[inline]{Number of rounds, Message cost?}
%\input{empirical}

\section{Discussion}
\label{sec:discussion}
Our algorithm, as described above, works under certain critical assumptions, some of which are described in detail below, in addition to the limitations discussed in Section~\ref{sec:limitations}. These assumptions are from the point of view of the attacker itself than the attack packets per se. We discuss them here to emphasize on the interesting open problems that these assumptions give rise to. For other areas of improvement, the reader is encouraged to refer to our discussion of future work in the next section and the answers to frequently encountered questions in the Appendix.
\paragraph{DoS from the Server during learning} The gateway can itself be subject to a denial of service attack from the server who is honestly following the protocol if the latter keeps sending more and more examples for training (on fake or false attack packets). Setting $r_{max}$ too high and $\varepsilon$ too low will be able to establish this without any regard to the reward (since the server can always terminate the protocol anytime it wishes). Currently, our protocol assumes that the server only runs this protocol when it \emph{is truly} under an attack and that it wants to stop the attack packets as soon as possible. In other words, we model the server as rational with respect to attack prevention and not malicious. Dealing with a malicious server/gateway is yet another interesting direction to explore.
\paragraph{Same Gateway Multiple Attacks}
What if another attack from the same gateway comes? Should the server run the protocol again, only this time with different attack packets? If yes, how many times should this be allowed? One solution is to adopt a strategy that if a gateway seems to be a source of too many attacks, then the server blocks all communication from that gateway. For example, the server gives only, say $3$, chances to the gateway before it stops receiving any packets from it at all. This seems to be a practical solution, but if the gateway is honest, the attacker can deliberately make the gateway fall victim to server's blacklisting. We believe that our spoof checking mechanism from Section~\ref{sec:remoteGate} will be able to prevent such a scenario if we give only a limited number of chances to the gateway before blacklisting it. However, there is a big room for improvement here to handle the different ways the attacker can act in this case.
\paragraph{Multiple Gateways Same Attack} Contrary to the attack above, yet another strategy for the attacker is to launch the same attack but from multiple sources, so that the server runs our protocol with multiple gateways and ends up paying more than what it would to a single gateway. One way to deal with this is for the server to locate a gateway that collects packets from these different \emph{source} gateways and run the protocol with it. Currently, this problem seems to be non-trivial if the gateway to be found is required to be different from the server itself. More specifically, if we view the incoming connections to the server as a graph, where the nodes are all nodes that can route packets to the server (or to other nodes who can), then this problem will involve finding the farthest node from the server that receives packets from each of these sources of attacks. The challenge is to do with without having any \emph{a priori} access to this graph. We believe that this is in itself an interesting open problem. 

\section{Future Work}
\label{sec:future}
In this section, we highlight some interesting open problems (in addition to the ones mentioned in the previous section) that can further improve \remoteGate and provide Internet scale security in the future. This is also in addition to what was previously mentioned in Section~\ref{sec:limitations}.

\paragraph{Limited attack packets}\label{limPackets} We assume that the number of attack packets and good packets with the server are enough so that the model learnt by the gateway is sufficiently accurate (wrt to the value of $\varepsilon$ fixed by the server). Further, we assume that the choice of the server in deciding training and test examples along with the order in which these will be sent to the gateway is supportive of the learning in the manner described above. We make these  assumptions because the server has no good estimate of how complex the learning problem here is and how many rounds it will need to achieve its tolerance setting for $\varepsilon$. We expect that the results in statistical learning theory may provide useful insights in trying to optimize the decisions here.

\paragraph{Bound on gateway's computation}  As mentioned before, we assume that the gateway can only perform polynomial (in the number of example packets sent by the server) number of computation steps in every round. However, similar to above, this raises the question of how one can guarantee that the gateway will be able to learn a model in at most $r_{max}$ rounds, given the examples from the server. Again, we seek the expertise of the statistical learning theorists in providing good estimates for the values of $r_{max}$ and $\varepsilon$ for which this assumption holds with high probability.  

\paragraph{Single source attack}  Our protocol, in its current form, allows the server to pay a gateway that provides this attack-packet filtering service. However, one may question the feasibility of this solution in case multiple sources exist for the same type of attack, each being behind different gateways. This attack may be coordinated or coincidental. In any case, if we allow the server the server to pay each gateway separately for this attack by running the protocol individually with each gateway, the server ends up paying a very high price for curbing the attack. At this point, we require the server to make a decision of whether it is more cost-effective to change its own firewall settings or engage into multiple remote services. For now, we assume that the server takes this decision wisely and leave it for future research to handle this problem more efficiently.

\paragraph{Integration with Distributed Ledger Technologies} One way to remove the assumption of a trusted third party for conflict resolution is to leverage the strong security guarantees that blockchains provide. The current smart contract model that Ethereum and other ledger technologies provide form an ideal candidate for this. A \remoteGate smart contract which handles the initial reward negotiation, reward aggregation over the learning phase and provides an inbuilt escrow service for the entire duration of attack prevention seems to be a promising solution towards a completely decentralized attack prevention that also takes advantage of the tremendous research that is trying to make cryptocurrencies more secure and scalable. We envision an implementation of such a smart-contract based global SDN in our future work.

\section{Conclusion}
\label{sec:conclusion}
In this paper, we introduced the high level idea behind global software defined networking to help prevent attacks closer to their point of origin as opposed to the conventional approach of installing firewalls and antivirus at the victim's end. We present a candidate algorithm for this, which we call \remoteGate, through which we envision enabling a server under attack to help configure the firewall of a remote gateway that was suspected to be the source of the attack packets. We designed \remoteGate to be an incentive-compatible protocol in which the server interactively helps the gateway to learn a model that, when deployed in the gateway's firewall, will filter out the attack packets and prevent them from reaching the server. We also highlighted some challenges and assumptions of our work and provided ideas for future research in this direction.

%\nocite{*}
\printbibliography

\appendix
\label{sec:appendix}
\newcommand{\question}[2]{\item[\textbf{Q#1.}]  \textbf{#2}}
\newcommand{\answer}[1]{\item[\textit{Ans.}]  #1}

\section{Frequently Asked Questions}
\subsection{Parameter setting}
\label{app:paramter}
\begin{enumerate}
	\question{1}{How are $v_S$ and $v_G$ selected for \remoteGate?}
	\answer{We assume human assistance in setting these parameters for now. In the longer vision of an automated system, these values can be set to appropriate functions of the attack detection mechanisms deployed at the server's end and the resource-usage monitoring systems deployed at the gateway's end.}
	
	\question{2}{How can the server make sure the gateway will be able to learn a model within $\epsilon$ accuracy within $r_\text{max}$ rounds?}
	\answer{Technically speaking, it cannot. The best the server can do it hope for such a learning to happen. However, one way to handle this is to be flexible in the choice of $\epsilon$. Once the server has completed some number of rounds with the gateway, the accuracy at the end of these rounds dictates how much fraction of the attack packets will be filtered if the model was deployed as it is. If the server determines that it is not in the best interest to continue any further given the nature of the attack, it can preempt the learning phase and ask for deployment. Ofcourse, taking such a decision in an automated manner may be challenging, which we leave for future work to handle.}
	
	\question{3}{How to set $\Gamma_\text{service}^S$ and $\Gamma_\text{service}^G$?}
	\answer{The service period sought by the server depends on the type of the attack. An ideal decision would be to stop the attack for long enough period for the server to be able to take a more affirmative action during the time (like giving an external investigation enough time to identify the real source of the attack and take action). It can also depend on the amount of funds available at the server along with some human provided parameters. For the gateway, this period can be a function of the amount of computational resources it thinks the filtering will incur along with the existing service periods with other servers on the Internet. In either case, optimizing the service time based on the type of attack is a topic of further investigation and research.}
\end{enumerate}

\subsection{Other discussion}
\label{app:technical}

\begin{enumerate}
	\question{4}{Where does the server get the money to pay the gateway?}
	\answer{This is similar to how payments are currently proposed to happen through IoT devices. The user can register a credit card or some digital payment mechanism securely on the server, which it uses to issue rewards to the gateways it interacts with. Although the exact implementation of such a system is beyond the scope of the paper, the server will be able to access the funds of the user operating it through any secure currency interchange (including the modern cryptocurrencies).}
	
	\question{5}{How can the gateway determine if the attack is invalid (not a real attack)?}
	\answer{Philosophically, there is no such way apart from matching the attack packets to some previously well known types of attacks. However, \remoteGate is more general in the sense that any packet that the server wishes not to receive can be labelled as \emph{attack}. Thus, our use of the word attack here is in a much broader sense, which in a way disables the gateway to question the authenticity of the labelling provided by the server.}
	
	\question{6}{Can we protect against a server who fakes attack packets (e.g. prevent access to , say Google, to someone)?}
	\answer{The current approach for \remoteGate is to only entertain requests from the server that block traffic to itself and not some other destination. Of course, in future this will be extended to general attack prevention, when this problem must be carefully looked into before the learning begins at the gateway's end. However, in a third possibility where someone else spoofs the servers IP address to block the traffic, then gateway's response will go to the servers IP address and not the spoofer. This is because the protocol is interactive. Hence, the spoofer will not be able to continue \remoteGate to the end.}
	
	\question{7}{Can \remoteGate deal with man-in-the-middle attack?}
	\answer{A possible man-in-the-middle attack for \remoteGate is when some adversary interacts with the server and the gateway making them believe that they are interacting with each other. One reason for such an intervention might be to obtain the reward the server has to pay the gateway. From the server's perspective, if the money goes to anybody other than the gateway, then the server can appeal to the authenticity of the secure payment channel that we assume exists in this case. Our proposed use of cryptocurrency here can help prevent this problem due to ledger transparency. Moreover, if the middle man fails to stop the attack packets after the payment is made, he is in a way exposing himself to the conflict resolution that the server will resort to in this scenario. Of course, \remoteGate is still in a very nascent stage and clever forms of attacks can be designed to break the system. We leave it to future work to optimize in these cases.}
	
	\question{8}{Does \remoteGate introduce new attacks?}
	\answer{Short answer, yes. More specifically, among many other ways, the attacker now collude with the gateway to get money in ways that will incentivize the gateway to collude with the adversary. However, such a behavior is hard to detect and even completely avoid. A real life analogy may help understand our philosophy here: upon calling the Police during an emergency, we lay faith on the fact the Police forces are not colluding with the attacker and are determined to help us. If we assume collusion by default, then no trust in such forces can be laid. Similarly, our approach is to trust the gateway by default and take appropriate actions when this trust is broken.}
	
	\question{9}{Is it possible for an adversary to tamper with the existing filtering of the gateway through careful design of the attack packets?}
	\answer{The answer to this question depends not on the design of the attack packets but on the classification of these packets as \emph{attacks} by the server. Even if the adversary spends a lot of resources to carefully design the attack packets, it needs to convince the server that these packets are indeed attack packets for which \remoteGate must be launched. An interesting question here is to compare the reward that the adversary collects (assuming \remoteGate is launched) vs the computation resources it spent in designing and transmitting those packets. If the latter is higher, then it becomes interesting to understand the situations under which the adversary is still incentivized to do so. We present this as an interesting future work to perform a resource competitive analysis of simulating this attack on the system.}
\end{enumerate}

\section{List of notations}
Fig.~\ref{fig:notations} provides a detailed list of notations.
\begin{figure*}
\begin{center}
	\begin{tabular}{ r p{1cm} p{12cm} }
		$\mathcal{A}$ && Set of attack packets with the server. \\
		$\Delta$ && Average delay between two attack packets received by the server. \\
		$\varepsilon$ && Server's error tolerance parameter in $(0,1)$ for the model learned by the gateway. \\
		$\texttt{LOG}_S$ && A complete log of server's interactions and internal computations during a run of the \remoteGate protocol. An entry (containing the state change and timestamp) is made in this log every time the server's internal state changes. \\
		$\texttt{LOG}_G$ && A complete log of gateway's interactions and internal computations during a run of the \remoteGate protocol. An entry (containing the state change and timestamp) is made in this log every time the gateway's internal state changes. \\
		$v_S$ && Server's true valuation of how much it should pay the gateway per correctly classified example during the learning phase. \\
		$v_G$ && Gateway's true valuation of how much it should be paid by the server per correctly classified example during the learning phase. \\
		$\rho_{S,r}$ && Server's bid for the reward per correctly classified example during $r^{th}$ round of the learning phase. \\
		$\rho_{G,r}$ && Gateway's bid for the reward per correctly classified example during $r^{th}$ round of the learning phase. \\
		$\rho_r$ && Reward chosen (per correctly classified example) for the $r^{th}$ round of the learning phase. \\
		$\Gamma_\text{service}^S$ && Server's requirement for the period of time it requires the gateway to maintain the model in its firewall for. \\
		$\Gamma_\text{service}^G$ && Gateway's guarantee for the period of time it will maintain the model in its firewall. \\
		$\iota$ && Number of installments in which the total service fee will be paid after deployment. \\
		$\texttt{fee}$ && Fees paid per installment by the server after the model has been deployed. \\
		$\texttt{TRAIN}_r$ && Training set of examples issued by the server for gateway's (supervised) learning in round $r$. \\
		$\texttt{TEST}_r$ && Set of examples issued by the server to test the model learned by the gateway at the end of round $r$.\\
		$\texttt{acc}(\ )$ && The number of test examples that were correctly classified by the gateway using its model, compared to the true labels with the server.\\
		$\texttt{LABEL}_S(\ )$ && The set of true labels for the examples. \\
		$\texttt{LABEL}_G(\ )$ && The set of labels for the examples as obtained by the gateway using its model. \\
		$\texttt{REWARD}_r$ && The total reward collected up to round $r$ that is to be paid to the gateway after the learning phase is complete.\\
		$\texttt{MODEL}_r$ && The (best effort polynomial-time) model learned by the gateway in round $r$ using all the labelled examples it has received from the server so far.\\
	\end{tabular}
	\caption{List of notations}
	\label{fig:notations}
\end{center}
\end{figure*}

\end{document}